\pdfoutput=1
\documentclass[a4paper,aps,pra,reprint,floatfix,superscriptaddress,showpacs]{revtex4-1}
\usepackage{amssymb}
\usepackage{amsmath}
\usepackage{amstext}
\usepackage{amsthm}
\usepackage{bbm}
\usepackage[percent]{overpic}
\usepackage{microtype}
\newcommand{\aaa}{\text{a}}
\newcommand{\bbb}{\text{b}}
\newcommand{\co}{\text{c}}
\newcommand{\si}{\text{s}}
\newcommand{\tC}{\text{C}}
\newcommand{\tS}{\text{S}}
\newcommand{\one}{\mathbbm{1}}

\newcommand{\D}{\text{d}}
\newcommand{\E}{\text{e}}
\newcommand{\I}{\text{i}}
\newcommand{\Trace}{\text{Tr}}
\newcommand{\Order}{\mathcal{O}}
\newcommand{\etal}{\textit{et al. }}
\providecommand{\tfra}[2]{\tfrac{#1}{#2}}
\providecommand{\eq}[1]{Equation \eqref{#1}}
\providecommand{\eqsimple}[1]{\eqref{#1}}
\providecommand{\Eq}[1]{Equation \eqref{#1}}
\providecommand{\fig}[1]{Figure \ref{#1}}
\providecommand{\Fig}[1]{Figure \ref{#1}}
\providecommand{\sect}[1]{Section \ref{#1}}

\providecommand{\app}[1]{Appendix \ref{#1}}

\newtheorem{theo}{Theorem}
\newtheorem{lem}{Lemma}
\newtheorem{cor}{Corollary}

\begin{document}

\title{The maximally entangled symmetric state in terms of the geometric measure}

\author{Martin Aulbach}
\email{pyma@leeds.ac.uk}
\affiliation{Department of Physics, Graduate School of Science, The
  University of Tokyo, Tokyo 113-0033, Japan}
\affiliation{The School of Physics and Astronomy, University of Leeds,
  Leeds LS2 9JT, UK}
\affiliation{Department of Physics, University of Oxford, Clarendon
  Laboratory, Oxford OX1 3PU, UK}

\author{Damian Markham}
\email{markham@telecom-paristech.fr}
\affiliation{CNRS, LTCI, Telecom ParisTech, 37/39 rue Dareau, 75014
  Paris, France}
\affiliation{Department of Physics, Graduate School of Science, The
  University of Tokyo, Tokyo 113-0033, Japan}

\author{Mio Murao}
\email{murao@phys.s.u-tokyo.ac.jp}
\affiliation{Department of Physics, Graduate School of Science, The
  University of Tokyo, Tokyo 113-0033, Japan}
\affiliation{Institute for Nano Quantum Information Electronics, The
  University of Tokyo, Tokyo 113-0033, Japan}

\begin{abstract}
  The geometric measure of entanglement is investigated for
  permutation symmetric pure states of multipartite qubit systems, in
  particular the question of maximum entanglement.  This is done with
  the help of the Majorana representation, which maps an $n$ qubit
  symmetric state to $n$ points on the unit sphere.  It is shown how
  symmetries of the point distribution can be exploited to simplify
  the calculation of entanglement and also help find the maximally
  entangled symmetric state. Using a combination of analytical and
  numerical results, the most entangled symmetric states for up to 12
  qubits are explored and discussed. The optimization problem on the
  sphere presented here is then compared with two classical
  optimization problems on the $S^2$ sphere, namely T\'{o}th's problem
  and Thomson's problem, and it is observed that, in general, they are
  different problems.
\end{abstract}
\maketitle

\section{Introduction}\label{introduction}

Being the fundamental resource in a wide range of situations in
quantum information processing, entanglement is considered as a
`standard currency' for quantum information tasks, and it is highly
desirable to know which states of a given system exhibit a high or
maximal amount of entanglement \cite{Horodecki09}. When it comes to
multipartite states this question becomes complicated. There are
different \emph{types} of entanglement \cite{Dur00}, alongside which
there are many different ways to quantify entanglement, each of which
may capture a different desirable quality of a state as a resource.

In this work, the geometric measure of entanglement, a distance-like
entanglement measure \cite{Shimony95,Wei04}, will be investigated to
analyze maximally entangled multipartite states. There are several
incentives to consider this particular measure.  Firstly, it has a
broad range of operational interpretations: for example, in local
state discrimination \cite{Hayashi06}, additivity of channel
capacities \cite{Werner02} and recently for the classification of
states as resources for measurement-based quantum computation
(MBQC)\cite{Gross09,Nest07,Mora10}.  Another advantage of the
geometric measure is that, while other known entanglement measures are
notoriously difficult to compute from their variational definitions,
the definition of the geometric measure allows for a comparatively
easy calculation. Furthermore, the geometric measure can be linked to
other distance-like entanglement measures, such as the robustness of
entanglement and the relative entropy of entanglement
\cite{Wei04,Hayashi08,Cavalcanti06}.  The function also has
applications in signal processing, particularly in the fields of
multi-way data analysis, high order statistics and independent
component analysis (ICA), where it is known under the name \emph{rank
  one approximation to high order tensors}
\cite{Lathauwer00,Zhang01,Kofidis02,Wang09,Ni07,Silva08}.

We focus our attention on permutation-symmetric states -- that is,
states that are invariant when swapping any pair of particles. This
class of states has been useful for different quantum information
tasks (for example, in leader election \cite{Dhondt06}).  It includes
the Greenberger-Horne-Zeilinger (GHZ) states \cite{Greenberger90}, W
states and Dicke states \cite{Dicke54}, and also occurs in a variety
of situations in many-body physics. There has been lots of activity
recently in implementing these states experimentally
\cite{Prevedel09,Wieczorek09}.  Furthermore, the symmetric properties
make them amenable to the analysis of entanglement properties
\cite{Hayashi08,Markham10,Toth09,Hubener09,Bastin09,Mathonet10}.

An important tool in this work will be the Majorana representation
\cite{Majorana32}, a generalization of the Bloch sphere representation
of single qubits, where a permutation-sym\-me\-tric state of $n$
qubits is unambiguously mapped to $n$ points on the surface of the
unit sphere.  Recently, the Majorana representation has proved very
useful in analyzing entanglement properties of symmetric states
\cite{Bastin09,Mathonet10,Markham10}. In particular, the geometric
measure of entanglement has a natural interpretation, and the Majorana
representation facilitates exploitation of further symmetries to
characterize entanglement \cite{Markham10}.  For example, the
two-qubit symmetric Bell state $| \psi^{+} \rangle = 1 / \sqrt{2}
\left( |01\rangle + |10\rangle \right)$ is represented by an antipodal
pair of points: the north pole $|0\rangle$ and the south pole
$|1\rangle$. Roughly speaking, symmetric states with a high degree of
entanglement are represented by point distributions that are well
spread out over the sphere.  We will use this idea along with other
symmetry arguments to look for the most entangled states. Along the
way we will compare this problem to other optimization problems of
point distributions on the sphere.

The paper is organized as follows. In \sect{geometric_measure}, the
definition and properties of the geometric measure of entanglement are
briefly recapitulated, which is followed by an introduction and
discussion of symmetric states in \sect{positive_and_symm}.  In
\sect{majorana_representation}, the Majorana representation of
symmetric states is introduced.  The problem of finding the maximally
entangled state is phrased in this manner, and is compared to two
other point distribution problems on $S^2$: T\'{o}th's problem and
Thomson's problem.  In \sect{analytic}, some theoretical results for
symmetric states are derived with the help of the intuitive idea of
the Majorana representation.  The numerically determined maximally
entangled symmetric states of up to 12 qubits are presented in
\sect{maximally_entangled_symmetric_states}.  Our results are
discussed in \sect{discussion}, and \sect{conclusion} contains the
conclusion.

\section{The geometric measure of entanglement}
\label{geometric_measure}

The geometric measure of entanglement is a distance-like entanglement
measure for pure multipartite states that assesses the entanglement of
a state in terms of its remoteness from the set of separable states
\cite{Vedral98}.  It is defined as the maximal overlap of a given pure
state with all pure product states \cite{Shimony95,Wei03,Barnum01} and
is also defined as the geodesic distance with respect to the
Fubini-Study metric \cite{Brody01}.  Here we present it in the inverse
logarithmic form of the maximal overlap, which is more convenient in
relation to other entanglement measures:
\begin{equation}\label{geo_1}
  E_{\text{G}}(| \psi \rangle ) = \min_{| \lambda \rangle \in
    \mathcal{H}_{\text{SEP}} } \log_2 \left(
    \frac{1}{ \vert \langle \lambda | \psi \rangle \vert^2 }
  \right) \enspace .
\end{equation}
$E_{\text{G}}$ is non-negative and zero iff $| \psi \rangle$ is a
product state.  We denote a product state closest to $| \psi \rangle$
by $| \Lambda_{\psi} \rangle \in \mathcal{H}_{\text{SEP}}$, and it
should be noted that a given $| \psi \rangle$ can have more than one
closest product state.  Indeed, we will usually deal with entangled
states that have several closest product states.  Due to its
compactness, the normalized, pure Hilbert space of a
finite-dimensional system (e.g. $n$ qudits) always contains at least
one state $| \Psi \rangle$ with maximal entanglement, and to each such
state relates at least one closest product state.  The task of
determining maximal entanglement can be formulated as a
max-min problem, with the two extrema not necessarily being
unambiguous:
\begin{equation}\label{geo_meas}
  \begin{split}
    E_{\text{G}}^{\text{max}} & = \max_{| \psi \rangle \in
      \mathcal{H}} \min_{| \lambda \rangle \in
      \mathcal{H}_{\text{SEP}} } \log_2 \left( \frac{1}{ \vert
        \langle \lambda | \psi \rangle \vert^2 } \right) \enspace , \\
    & = \max_{| \psi \rangle \in \mathcal{H}} \log_2 \left(
      \frac{1}{ \vert \langle \Lambda_{\psi} | \psi \rangle \vert^2 }
    \right) \enspace , \\
    & = \log_2 \left( \frac{1}{ \vert \langle \Lambda_{\Psi} | \Psi
        \rangle \vert^2 } \right) \enspace .
  \end{split}
\end{equation}
It is often more convenient to define $G(| \psi \rangle ) = \max_{|
  \lambda \rangle } \vert \langle \lambda | \psi \rangle \vert \, ,$
so that we obtain $E_{\text{G}} = \log_2 ( 1/ {G}^2 )$.  Because of
the monotonicity of this relationship, the task of finding the
maximally entangled state is equivalent to solving the min-max problem
\begin{equation}\label{minmax}
  \min_{| \psi \rangle \in \mathcal{H}}
  G(| \psi \rangle ) =
  \min_{| \psi \rangle \in \mathcal{H}}
  \max_{| \lambda \rangle \in \mathcal{H}_{\text{SEP}} }
  \vert \langle \lambda | \psi \rangle \vert \enspace .
\end{equation}
As mentioned in the Introduction, there are several advantages of this
measure of entanglement. First of all, it has several operational
interpretations. It has implications for channel capacity
\cite{Werner02} and can also be used to give conditions as to when
states are useful resources for MBQC \cite{Gross09,Nest07,Mora10}. If
the entanglement of a set of resource states scales anything below
logarithmically with the number of parties, it cannot be an efficient
resource for deterministic universal MBQC \cite{Nest07}. On the other
hand, somewhat surprisingly, if the entanglement is too large, it is
also not a good resource for MBQC. If the geometric measure of
entanglement of an $n$ qubit system scales larger than $n - \delta$
(where $\delta$ is some constant), then such a computation can be
simulated efficiently computationally \cite{Gross09}. Of course, we
should also note that there are many other quantum information tasks
that are not restricted by such requirements. For example, the
$n$-qubit GHZ state can be considered the most non-local with respect
to all possible two-output, two-setting Bell inequalities
\cite{Werner01}, whereas the geometric measure is only $E_{\text{G}}
(|\text{GHZ}\rangle)=1$, independent of $n$. The $n$-qubit W state, on
the other hand, is the optimal state for leader election
\cite{Dhondt06} with entanglement $E_{\text{G}} (| \text{W} \rangle )
= \log_2 (n/(n-1))^{n-1}$.  Indeed, for local state discrimination,
the role of entanglement in blocking the ability to access information
locally is strictly monotonic -- the higher the geometric measure of
entanglement, the harder it is to access information locally
\cite{Hayashi06}.

In addition, the geometric measure $E_{\text{G}}$ has close links to
other distance-like entanglement measures, namely the (global)
robustness of entanglement $R$ \cite{Vidal99} and the relative entropy
of entanglement $E_{\text{R}}$ \cite{Vedral98}.  Between these
measures the inequalities $E_{\text{G}} \leq E_{\text{R}} \leq \log_2
(1 + R)$ hold for all states \cite{Wei04,Hayashi08,Cavalcanti06}, and
they turn into equalities for stabilizer states (e.g. GHZ state),
Dicke states (e.g. W state) and permutation-antisymmetric basis states
\cite{Hayashi06,Hayashi08,Markham07}.

An upper bound for the entanglement of pure $n$ qubit states is given
in \cite{Jung08} as
\begin{equation}
  E_{\text{G}} ( | \psi \rangle ) \leq n-1 \enspace .
\end{equation}
We can see that this allows for states to be more entangled than is
useful, e.g. for MBQC. Indeed, although no states of more than two
qubits reach this bound \cite{Jung08}, most states of $n$ qubits have
entanglement $E_{\text{G}} > n - 2 \log_2 (n) - 3$ \cite{Gross09}.  In
the next section, we will see that symmetric states have generally
lower entanglement.

We can also make a general statement for positive states that will
help us in calculating entanglement for this smaller class of
states. For finite-dimensional systems, a general quantum state can be
written in the form $| \psi \rangle = \sum_i a_i | i \rangle$ with an
orthonormalized basis $\{ | i \rangle \}$ and complex coefficients
$a_i \in \mathbb{C}$.  We will call $| \psi \rangle$ a \emph{real
  state} if -- for a given basis $\{ | i \rangle \}$ -- all
coefficients are real ($a_i \in \mathbb{R}$), and likewise call $|
\psi \rangle$ a \emph{positive state} if the coefficients are all
positive ($a_i \geq 0$). A \emph{computational basis} is one made up
of tensors of local bases.

\begin{lem}\label{lem_positive}
  Every state $| \psi \rangle$ of a finite-dimensional system that is
  positive with respect to some computational basis has at least one
  positive closest product state $| \Lambda_{\psi} \rangle$.
\end{lem}
\begin{proof}
  Picking any computational basis in which the coefficients of $| \psi
  \rangle$ are all positive, we denote the basis of subsystem $j$ with
  $\{ | i_{j} \rangle \}$, and can write the state as $| \psi \rangle
  = \sum_{\vec{i}} a_{\vec{i}} \, | i_1 \rangle \cdots | i_n \rangle$,
  with $\vec{i} = ( i_1 , \dots , i_n)$ and $a_{\vec{i}} \geq 0$.  A
  closest product state of $| \psi \rangle$ can be written as $|
  \Lambda_{\psi} \rangle = \bigotimes_{j} | \sigma_{j} \rangle$, where
  $| \sigma_{j} \rangle = \sum_{i_j} b^{j}_{i_j} | i_j \rangle$ (with
  $b^{j}_{i_j} \in \mathbb{C}$) is the state of subsystem $j$. Now
  define a new product state with positive coefficients as $|
  \Lambda_{\psi} ' \rangle = \bigotimes_{j} | \sigma_{j} ' \rangle$,
  where $| \sigma_{j} ' \rangle = \sum_{i_j} | b^{j}_{i_j} | \, | i_j
  \rangle$.  Because of $|\langle \psi | \Lambda_{\psi} ' \rangle| =
  \sum_{\vec{i}} a_{\vec{i}} \prod_{j} | b^{j}_{i_j} | \geq \left|
    \sum_{\vec{i}} a_{\vec{i}} \prod_{j} b^{j}_{i_j} \right| =
  |\langle \psi | \Lambda_{\psi} \rangle|$, the positive state $|
  \Lambda_{\psi} ' \rangle$ is a closest product state of $| \psi
  \rangle$.
\end{proof}

This lemma, which was also shown in \cite{Zhu10}, asserts that
positive states have at least one positive closest product state, but
there can nevertheless exist other non-positive closest product
states.  A statement analogous to Lemma \ref{lem_positive} does not
hold for real states, and it is easy to find examples of real states
that have no real closest product state.

From now on we will simply denote entanglement instead of referring to
the geometric measure of entanglement. It must be kept in mind,
however, that the maximally entangled state of a multipartite system
subtly depends on the chosen entanglement measure \cite{Plenio07}.

\section{Permutation symmetric states}\label{positive_and_symm}

In general it is very difficult to find the closest product state of a
given quantum state, due to the large amount of parameters in $|
\Lambda \rangle$. The problem will be considerably simplified,
however, when considering permutation-symmetric states.  In
experiments with many qubits, it is often not possible to access
single qubits individually, necessitating a fully symmetrical
treatment of the initial state and the system dynamics \cite{Toth07}.
The ground state of the Lipkin-Meshkov-Glick model was found to be
permutation-invariant, and its entanglement was quantified in term of
the geometric measure and its distance-related cousins \cite{Orus08}.
For these reasons it is worth analyzing various theoretical and
experimental aspects of the entanglement of symmetric states, such as
entanglement witnesses or experimental setups
\cite{Korbicz05,Korbicz06}.

The symmetric basis states of a system of $n$ qubits are given by the
Dicke states, the simultaneous eigenstates of the total angular
momentum $J$ and its $z$-component $J_z$
\cite{Dicke54,Stockton03,Toth07}.  They are mathematically expressed
as the sum of all permutations of computational basis states with
$n-k$ qubits being $|0\rangle$ and $k$ being $|1\rangle$.
\begin{equation}\label{dicke_def}
  | S_{n,k} \rangle = {\binom{n}{k}}^{- 1/2} \sum_{\text{perm}} \;
  \underbrace{ | 0 \rangle | 0 \rangle \cdots | 0 \rangle }_{n-k}
  \underbrace{ | 1 \rangle | 1 \rangle \cdots | 1 \rangle }_{k}
  \enspace ,
\end{equation}
with $0 \leq k \leq n$, and where we omitted the tensor symbols that
mediate between the $n$ single qubit spaces.  The Dicke states
constitute an orthonormalized set of basis vectors for the symmetric
Hilbert space $\mathcal{H}_{\text{s}}$.  The notation $| S_{n,k}
\rangle$ will sometimes be abbreviated as $| S_{k} \rangle$ when the
number of qubits is clear.

Recently, there has been a very active investigation into the
conjecture that the closest product state of a symmetric state is
symmetric itself \cite{Wei04,Hayashi08,Hubener09}. A proof of this
seemingly straightforward statement is far from trivial, and after
some special cases were proofed \cite{Hayashi09,Wei10}, H\"{u}bener
\etal \cite{Hubener09} were able to extend this result to the general
case.  They also showed that, for $n \geq 3$ qudits (general quantum
$d$-level systems) the closest product state of a symmetric state is
\emph{necessarily} symmetric.  This result greatly reduces the
complexity of finding the closest product state and thus the
entanglement of a symmetric state.

A general pure symmetric state of $n$ qubits is a linear combination
of the $n+1$ symmetric basis states, with the operational nature being
that the state remains invariant under the permutation of any two of
its subsystems.

A closest product state of $| S_{n,k} \rangle$ is \cite{Hayashi08}
\begin{equation}\label{dicke_cs}
  | \Lambda \rangle = \Big( \sqrt{ \tfra{n-k}{n} } \, |0\rangle +
  \sqrt{ \tfra{k}{n} } \, |1\rangle \Big)^{\otimes n} \enspace ,
\end{equation}
i.e. a tensor product of $n$ identical single qubit states.  From
this, the amount of entanglement is found to be
\begin{equation}\label{dicke_ent}
  E_{\text{G}} ( | S_{n,k} \rangle ) = \log_2 \left(
    \frac{ \big( \frac{n}{k} \big)^k \big( \frac{n}{n-k}
      \big)^{n-k}} {\binom{n}{k}} \right) \enspace .
\end{equation}
This formula straightforwardly gives the maximally entangled Dicke
state. For even $n$ it is $| S_{n,n/2} \rangle$ and for odd $n$ the
two equivalent states $| S_{n,(n+1)/2} \rangle$ and $| S_{n,(n-1)/2}
\rangle$. In general, however, the maximally entangled symmetric state
of $n$ qubits is a superposition of Dicke states. Nevertheless,
\eq{dicke_ent} can be used as a lower bound to the maximal
entanglement of symmetric states.  This bound can be approximated by
the Stirling formula for large $n$ as $E_{\text{G}} \geq \log_2
\sqrt{n \pi/2}$.

An upper bound to the geometric measure for symmetric $n$ qubit states
can be easily found from the well-known decomposition of the identity
on the symmetric subspace (denoted $\one_{\text{Symm}}$, see
e.g. \cite{Renner}),
\begin{equation}
  \int_{\mathcal{S}(\mathcal{H})}
  (|\theta\rangle\langle\theta|)^{\otimes n}\omega(\theta)
  = \frac{1}{n+1}\one_{\text{Symm}} \enspace ,
\end{equation}
where $\omega$ denotes the uniform probability measure over the unit
sphere $\mathcal{S}(\mathcal{H})$ on Hilbert space $\mathcal{H}$.  We
can easily see that $G ( | \psi \rangle )^2 = \max_{\omega \in
  \mathcal{H}_{\text{SEP}}} \Trace (\omega |\psi\rangle\langle
\psi|)\geq 1 / (n+1)$. Hence, for any symmetric state of $n$ qubits,
the geometric measure of entanglement is upper bounded by
\begin{equation}
  E_{\text{G}} (|\psi\rangle_{\text{s}}) \leq \log_2 (n+1) \enspace .
\end{equation}
An alternative proof that has the benefit of being visually accessible
is presented in \app{normalization_bloch}.

The maximal symmetric entanglement for $n$ qubits thus scales
polylogarithmically between $\Order (\log \sqrt{n})$ and $\Order (\log
n)$.  To compare this with the general non-sym\-me\-tric case,
consider the lower bound of the maximal $n$ qubit entanglement ($n$
even) of $E_{\text{G}} \geq (n/2)$ \footnote{A trivial example of an
  $n$ qubit state with $E_{\text{G}} = (n/2)$ are $(n/2)$ bipartite
  Bell states, each of which contributes 1 ebit. Another example is
  the 2D cluster state of $n$ qubits which has $E_{\text{G}} = (n/2)$
  \protect{\cite{Markham07}}.}.  Thus the maximal entanglement of
general states scales much faster, namely linearly rather than
logarithmically. As mentioned, for most states the entanglement is
even higher and thus too entangled to be useful for MBQC. While the
bounds for symmetric states mean that permutation-symmetric states are
never too entangled to be useful for MBQC, unfortunately their scaling
is also too low to be good universal deterministic resources
\cite{Nest07}. They may nevertheless be candidates for approximate,
stochastic MBQC \cite{Mora10}. Regardless of their use as resources
for MBQC, the comparatively high entanglement of symmetric states
still renders them formidable candidates for specific quantum
computations or as resources for other tasks, such as the leader
election problem \cite{Dhondt06} and LOCC discrimination
\cite{Hayashi06}.

We end this section by mentioning a simplification with respect to
symmetric positive states.  States that are symmetric as well as
positive in some computational basis are labelled as \emph{positive
  symmetric}. From the previous discussion it is clear that such
states have a closest product state which is positive symmetric
itself, a result first shown in \cite{Hayashi08}.  It should be noted
that, while each closest product state of a positive symmetric state
is \emph{necessarily} symmetric for $n \geq 3$ qudits, it \emph{need
  not} be positive. We can formulate this as a statement akin to Lemma
\ref{lem_positive}.

\begin{lem}\label{lem_positive_symmetric}
  Every symmetric state $| \psi \rangle_{\text{s}}$ of $n$ qudits,
  which is positive in some computational basis, has at least one
  positive symmetric closest product state $| \Lambda_{\psi}
  \rangle_{\text{s}}$.
\end{lem}

\section{Majorana representation of symmetric states}
\label{majorana_representation}

With the discussion of the geometric measure and symmetric states
behind us, we have gathered the prerequisites to introduce a crucial
tool, the Majorana representation. It will help us to understand the
amount of entanglement of symmetric states.

\subsection{Definition}\label{majorana_definition_sect}

In classical physics, the angular momentum $\mathbf{J}$ of a system
can be represented by a point on the surface of the 3D unit sphere
$S^2$, which corresponds to the direction of $\mathbf{J}$.  No such
simple representation is possible in quantum mechanics, but Majorana
\cite{Majorana32} pointed out that a pure state of spin-$j$ can be
uniquely represented by $2j$ not necessarily distinct points on $S^2$.
This is a generalization of the spin-$1/2$ (qubit) case, where the 2D
Hilbert space is isomorphic to the unit vectors on the Bloch sphere.

An equivalent representation also exists for per\-mu\-tation-symmetric
states of $n$ spin-$1/2$ particles \cite{Majorana32,Bacry74}.  By
means of this `Majorana representation' any symmetric state of $n$
qubits $| \psi \rangle_{\text{s}}$ can be uniquely composed from a sum
over all permutations $P : \{ 1, \dots , n \} \rightarrow \{ 1, \dots,
n \}$ of $n$ undistinguishable single qubit states $\{ | \phi_1
\rangle, \dots , | \phi_n \rangle \}$:
\begin{align}
  | \psi \rangle_{\text{s}} = {} & \frac{1}{\sqrt{K}}
  \sum_{ \text{perm} } | \phi_{P(1)} \rangle | \phi_{P(2)} \rangle
  \cdots | \phi_{P(n)} \rangle \enspace , \label{majorana_definition} \\
  & \text{with} \quad | \phi_i \rangle = \cos \tfra{\theta_i}{2} \,
  |0\rangle + \E^{\I \varphi_i} \sin
  \tfra{\theta_i}{2} |1\rangle \enspace , \nonumber \\
  & \text{and} \quad \; K = n! \sum_{\text{perm}} \, \prod_{i = 1}^{n}
  \, \langle \phi_{i} | \phi_{ P(i) } \rangle \enspace . \nonumber
\end{align}
The normalization factor $K$ is in general different for different $|
\psi \rangle_\text{s}$.  By means of \eq{majorana_definition}, the
multi-qubit state $| \psi \rangle_\text{s}$ can be visualized by $n$
unordered points (each of which has a Bloch vector pointing in its
direction) on the surface of a sphere.  We call these points the
\emph{Majorana points} (MP), and the sphere on which they lie the
\emph{Majorana sphere}.

With \eq{majorana_definition}, the form of a symmetric state $| \psi
\rangle_{\text{s}}$ can be explicitly determined if the MPs are
known. If the MPs of a given state $| \psi \rangle_{\text{s}} =
\sum^{n}_{k=0} a_k | S_k \rangle$ are unknown, they can be determined
by solving a system of $n+1$ equations.
\begin{gather}
  a_k = {\binom{n}{k}}^{1/2} \sum_{ \text{perm} }
  \tS_{P(1)} \cdots \tS_{P(k)} \tC_{P(k+1)} \cdots \tC_{P(n)} \: ,
  \label{state_to_mp} \\
  \text{with} \quad
  \tC_{i} = \cos \tfra{\theta_i}{2} \enspace ,
  \qquad \tS_{i} = \E^{\I \varphi_{i}} \sin
  \tfra{\theta_i}{2} \enspace . \nonumber
\end{gather}

The Majorana representation has been rediscovered several times, and
has been put to many different uses across physics. In relation to the
foundations of quantum mechanics, it has been used to find efficient
proofs of the Kochen-Specker theorem \cite{Zimba93,PenroseRindler} and
to study the `quantumness' of pure quantum states in several respects
\cite{Zimba06,Giraud10}, as well as the approach to classicality in
terms of the discriminability of states \cite{Markham03}. It has also
been used to study Berry phases in high spin \cite{Hannay96} and
quantum chaos \cite{Hannay98,Leboeuf91}. Within many-body physics it
has been used for finding solutions to the Lipkin-Meshkov-Glick model
\cite{Ribiero08}, and for studying and identifying phases in spinor
BEC \cite{Barnett06,Barnett07,Barnett08,Makela07}.  It has also been
used to look for optimal resources for reference frame alignment
\cite{Kolenderski08} and for phase estimation \cite{Kolenderski09}.

Recently, the Majorana representation has also become a useful tool in
studying the entanglement of permutation-symmetric states. It has been
used to search for and characterize different classes of entanglement
\cite{Bastin09,Mathonet10,Markham10}, which have interesting mirrors
in the classification of phases in spinor condensates
\cite{Markham10,Barnett07}. Of particular interest, in this work, is
that it gives a natural visual interpretation of the geometric measure
of entanglement \cite{Markham10}, and we will see how symmetries in
the point distributions can be used to calculate the entanglement and
assist in finding the most entangled states.

The connection to entanglement can first be noticed by the fact that
the point distribution is invariant under local unitary maps.
Applying an arbitrary single-qubit unitary operation $U$ to each of
the $n$ subsystems yields the LU map
\begin{equation}\label{m_def_1}
  | \psi \rangle_{\text{s}} \; \longmapsto \;
  | \varphi \rangle_{\text{s}} \equiv U \otimes \cdots \otimes
  U \, | \psi \rangle_{\text{s}} \enspace ,
\end{equation}
and from \eq{majorana_definition} it follows that
\begin{align}\label{lusphere}
    | \varphi \rangle_{\text{s}} = {} & \frac{1}{\sqrt{K}} \sum_{
      \text{perm} } | \vartheta_{P(1)} \rangle | \vartheta_{P(2)}
    \rangle \cdots | \vartheta_{P(n)} \rangle \enspace , \\
   & \text{with} \quad | \vartheta_i \rangle = U | \phi_i \rangle \;
    \forall i \enspace . \nonumber
\end{align}
In other words, the symmetric state $| \psi \rangle_{\text{s}}$ is
mapped to another symmetric state $| \varphi \rangle_{\text{s}}$, and
the MP distribution of $| \varphi \rangle_{\text{s}}$ is obtained by a
joint rotation of the MP distribution of $| \psi \rangle_{\text{s}}$
on the Majorana sphere along a common axis.  Therefore $| \psi
\rangle_{\text{s}}$ and $| \varphi \rangle_{\text{s}}$ have different
MPs, but the same \emph{relative} distribution of the MPs, and the
entanglement remains unchanged.

When it comes to the geometric measure of entanglement, we can be even
more precise.  For $n \geq 3$ qubits, every closest product state $|
\Lambda \rangle_{\text{s}}$ of a symmetric state $| \psi
\rangle_{\text{s}}$ is symmetric itself \cite{Hubener09}, so that one
can write $| \Lambda \rangle_{\text{s}} = | \sigma \rangle^{\otimes
  n}$ with a single qubit state $| \sigma \rangle$, and visualize $|
\Lambda \rangle_{\text{s}}$ by the Bloch vector of $| \sigma
\rangle$. In analogy to the Majorana points, we refer to $| \sigma
\rangle$ as a \emph{closest product point} (CPP).

For the calculation of the geometric measure of entanglement, the
overlap with a symmetric product state $| \lambda \rangle = | \sigma
\rangle^{\otimes n}$ is
\begin{equation}\label{bloch_product}
  | \langle \lambda | \psi \rangle_{\text{s}} | =
  \frac{n!}{\sqrt{K}} \, \prod_{i=1}^{n} \, | \langle \sigma |
  \phi_i \rangle | \enspace .
\end{equation}
The task of determining the CPP of a given symmetric state is thus
equivalent to maximizing the absolute value of a product of scalar
products. From a geometrical point of view, the $\langle \sigma |
\phi_i \rangle$ are the angles between the two corresponding points on
the Majorana sphere, and thus the determination of the CPP can be
viewed as an optimization problem for a product of geometrical angles.

\subsection{Examples}\label{examples}

We will now demonstrate the Majorana representation for two and three
qubit symmetric states.  The case of two qubits is very simple,
because any distribution of two points can be rotated on the Majorana
sphere in a way that both MPs are positive, with $| \phi_1 \rangle = |
0 \rangle$ and $| \phi_2 \rangle = \cos \tfra{\theta}{2} | 0 \rangle +
\sin \tfra{\theta}{2} | 1 \rangle$ for some $\theta \in [0, \pi
]$. One CPP of this MP distribution is easily found to be $| \sigma
\rangle = \cos \tfra{\theta}{4} | 0 \rangle + \sin \tfra{\theta}{4} |
1 \rangle$. \Fig{bell_pic} shows two examples for $\theta = \pi / 2$
and $\theta = \pi$, with the latter representing the Bell state $|
\psi^{+} \rangle = 1 / \sqrt{2} \left( |01\rangle + |10\rangle
\right)$. Due to the azimuthal symmetry of its MPs on the sphere, the
CPPs form a continuous ring $| \sigma \rangle = 1 / \sqrt{2} \left( |
  0 \rangle + \E^{\I \varphi} | 1 \rangle \right)$, with $\varphi \in
[0,2 \pi)$ around the equator.  The amount of entanglement is
$E_{\text{G}} ( | \psi^{+} \rangle ) = 1$.  For two qubits, the
maximally entangled symmetric states are easily found to be those
whose MPs lie diametrically opposite on the sphere.

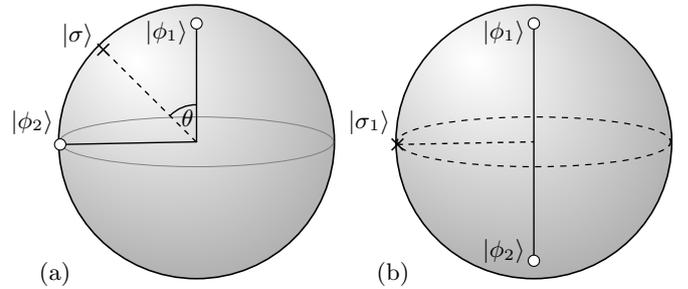
\begin{figure}
  \begin{center}
    \begin{overpic}[scale=.5]{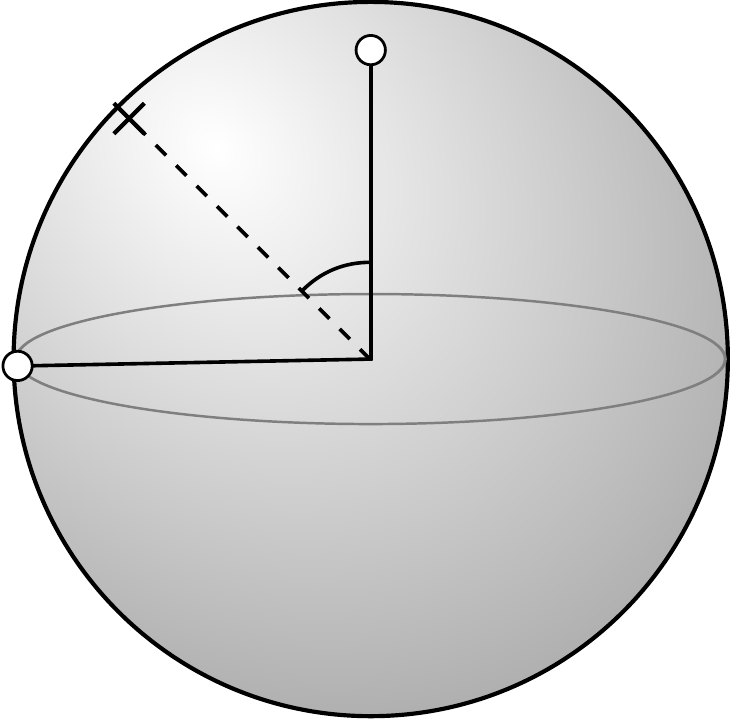}
      \put(-5,0){(a)}
      \put(32.5,86){$| \phi_1 \rangle$}
      \put(-15,54){$| \phi_2 \rangle$}
      \put(3,85){$| \sigma \rangle$}
      \put(45.3,54.7){$\theta$}
    \end{overpic}
    \hspace{5mm}
    \begin{overpic}[scale=.5]{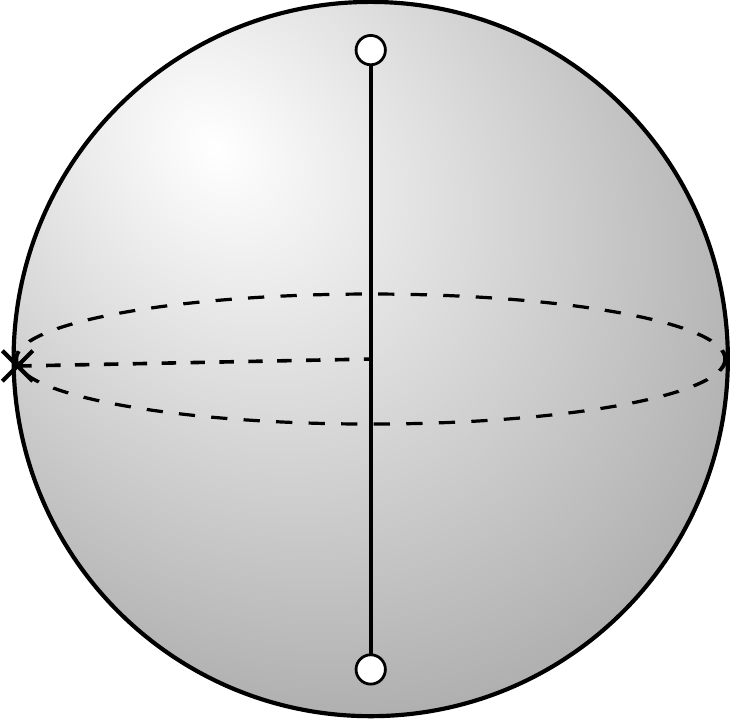}
      \put(-5,0){(b)}
      \put(32.5,86){$| \phi_1 \rangle$}
      \put(32.5,8){$| \phi_2 \rangle$}
      \put(-15,54){$| \sigma_1 \rangle$}
    \end{overpic}
  \end{center}
  \caption{\label{bell_pic} The Majorana representations of two
    symmetric states of two qubits.  MPs are shown as white circles
    and CPPs as dashed lines or crosses. Panel (a) depicts the state
    $\sqrt{2/3} \, | 00 \rangle + \sqrt{1/6} \, ( | 01 \rangle + | 10
    \rangle )$.  Its single CPP lies in the middle of the two MPs at
    an angle of $\theta = \pi / 4$. Panel (b) shows the Bell state $|
    \psi^{+} \rangle = 1 / \sqrt{2} \left( |01\rangle + |10\rangle
    \right)$, whose CPPs form a continuous ring on the equatorial
    belt.}
\end{figure}

For three qubit states, the GHZ state and the W state, both of which
are positive and symmetric, are considered as extremal among three
qubit states \cite{Tamaryan09}. The tripartite GHZ state $| \text{GHZ}
\rangle = 1 / \sqrt{2} \left( | 000 \rangle + | 111 \rangle \right)$
\cite{Greenberger90} has the MPs
\begin{equation}\label{GHZ-maj}
  \begin{split}
    | \phi_1 \rangle & = \tfra{1}{\sqrt{2}} \big( | 0 \rangle +
    | 1 \rangle \big) \enspace , \\
    | \phi_2 \rangle & = \tfra{1}{\sqrt{2}} \big( | 0 \rangle +
    \E^{\I 2 \pi / 3} | 1 \rangle \big) \enspace , \\
    | \phi_3 \rangle & = \tfra{1}{\sqrt{2}} \big( | 0 \rangle +
    \E^{\I 4 \pi / 3} | 1 \rangle \big) \enspace .
  \end{split}
\end{equation}
Its two CPPs are easily calculated to be $| \sigma_1 \rangle =
|0\rangle$ and $| \sigma_2 \rangle = |1\rangle$, yielding an
entanglement of $E_{\text{G}} ( | \text{GHZ} \rangle ) =
1$. \Fig{ghz_w_pic}(a) shows the distribution of the MPs and CPPs for
the GHZ state.  The three MPs form an equilateral triangle inside the
equatorial belt, and the two CPPs lie at the north and the south pole,
respectively.

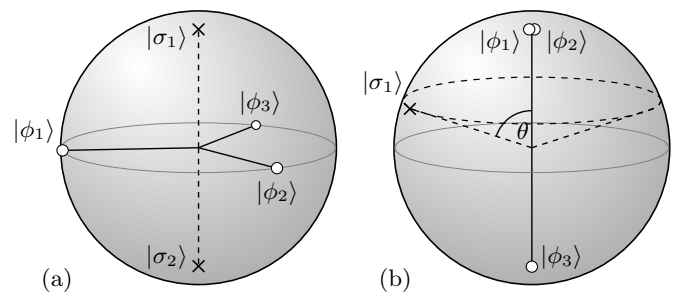
\begin{figure}[b]
  \begin{center}
    \begin{overpic}[scale=.5]{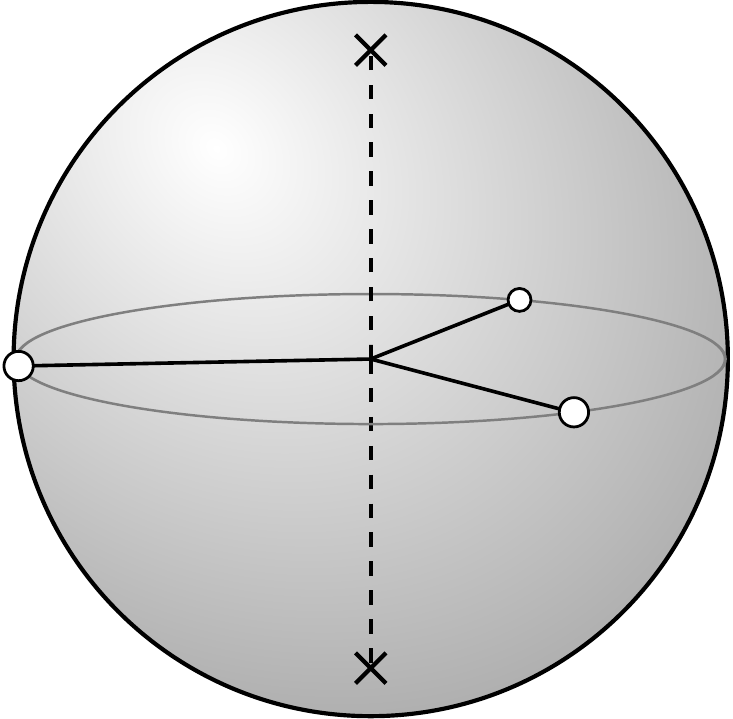}
      \put(-5,0){(a)}
      \put(-15,53){$| \phi_1 \rangle$}
      \put(71,31){$| \phi_2 \rangle$}
      \put(65,63){$| \phi_3 \rangle$}
      \put(32,86){$| \sigma_1 \rangle$}
      \put(32,8){$| \sigma_2 \rangle$}
    \end{overpic}
    \hspace{5mm}
    \begin{overpic}[scale=.5]{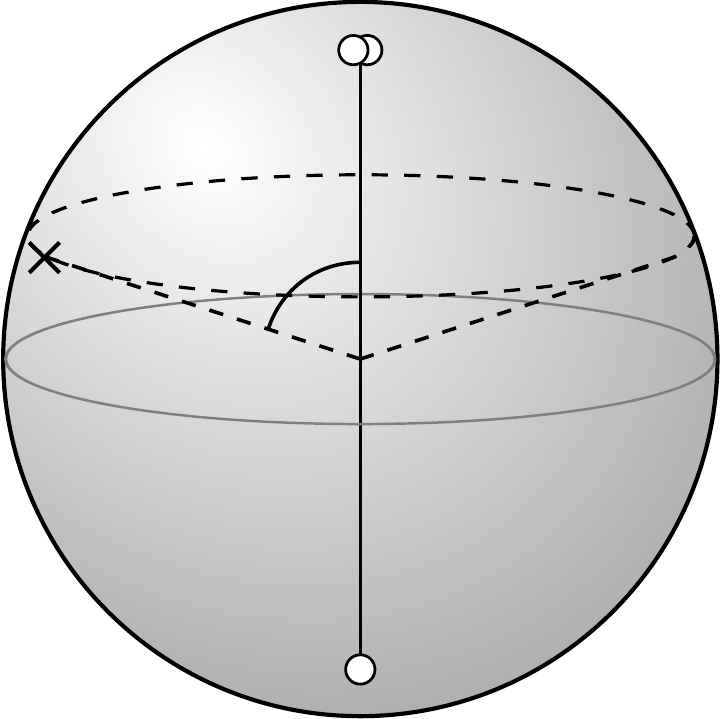}
      \put(-5,0){(b)}
      \put(31,86){$| \phi_1 \rangle$}
      \put(55,86){$| \phi_2 \rangle$}
      \put(53,8){$| \phi_3 \rangle$}
      \put(-11,71){$| \sigma_1 \rangle$}
      \put(44.5,53){$\theta$}
    \end{overpic}
  \end{center}
  \caption{\label{ghz_w_pic} The MPs and CPPs of the three qubit (a)
    GHZ state and (b) W state.  The GHZ state has two discrete CPPs
    whereas for the W state the CPPs form a continuous ring due to the
    azimuthal symmetry.}
\end{figure}

The W state $| \text{W} \rangle = | S_{3,1} \rangle = 1/ \sqrt{3} ( |
001 \rangle + | 010 \rangle + | 100 \rangle )$ is a Dicke state, and
its MPs can be immediately accessed from its form as
\begin{equation}\label{W-maj}
  \begin{split}
    | \phi_1 \rangle & = | \phi_2 \rangle = | 0 \rangle \enspace , \\
    | \phi_3 \rangle & = | 1 \rangle \enspace .
  \end{split}
\end{equation}
Generally, the definition \eqsimple{dicke_def} of the Dicke states $|
S_{n,k} \rangle$ asserts that $n-k$ MPs lie at the north pole and $k$
at the south pole.  \Eq{dicke_cs} yields $| \sigma_1 \rangle = \sqrt{
  2/3 } \, |0\rangle + \sqrt{ 1/3 } \, |1\rangle$ as a positive CPP of
the W state, and from the azimuthal symmetry of the MP distribution it
is clear that the set of all CPPs is formed by the ring of vectors $|
\sigma \rangle = \sqrt{2/3} \, | 0 \rangle + \E^{\I \varphi}
\sqrt{1/3} \, | 1 \rangle$, with $\varphi \in [0,2 \pi)$.
\Fig{ghz_w_pic}(b) shows the MPs and CPPs of $| \text{W} \rangle$. The
amount of entanglement is $E_{\text{G}} ( | \text{W} \rangle ) =
\log_2 \left( 9/4 \right) \approx 1.17$, which is higher than that of
the GHZ state.  It was recently shown that, in terms of the geometric
measure, the W state is the maximally entangled of all three qubit
states \cite{Chen10}.

\begin{figure}[b]
  \begin{center}
    \begin{minipage}{86mm}
      \begin{center}
        \begin{overpic}[scale=.22]{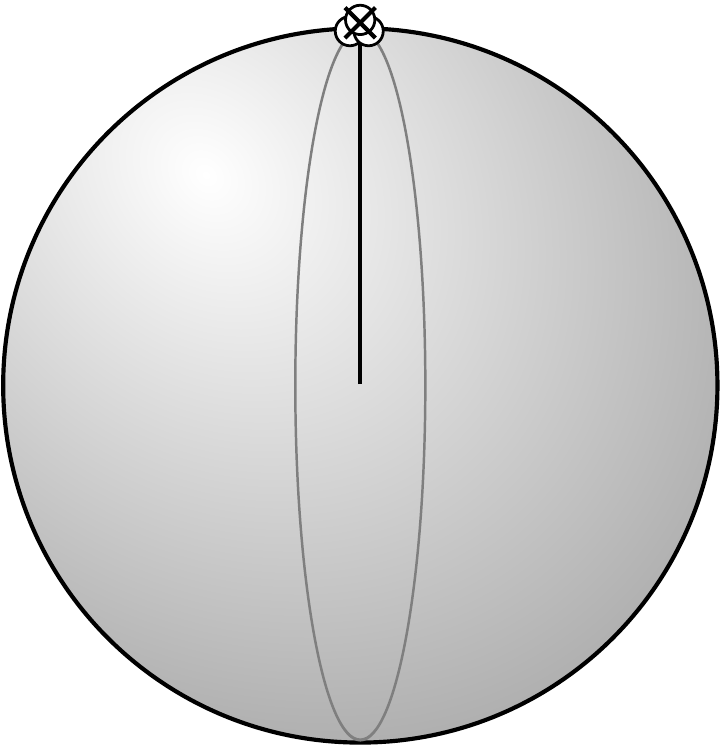} \put(-14,0){(a)}
        \end{overpic}
        \begin{overpic}[scale=.22]{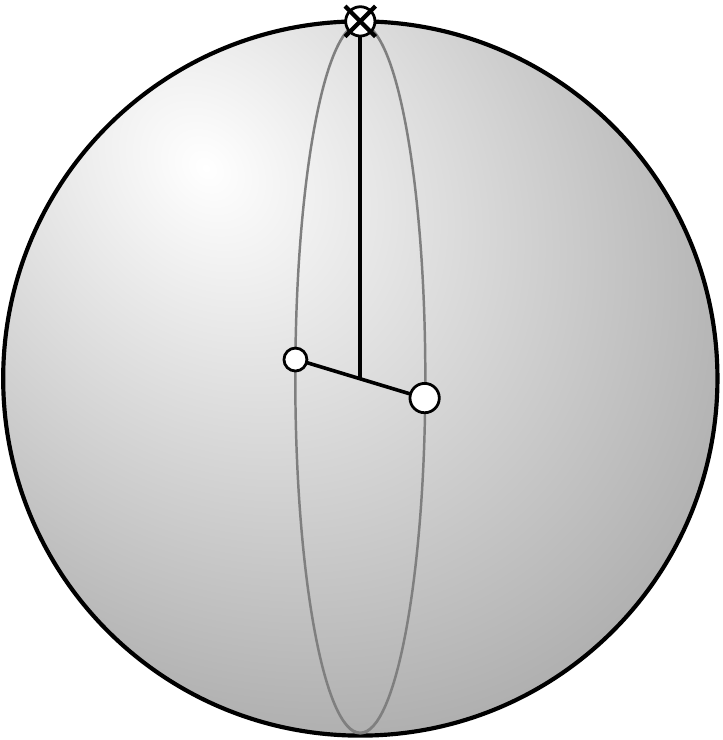} \put(-14,0){(b)}
        \end{overpic}
        \begin{overpic}[scale=.22]{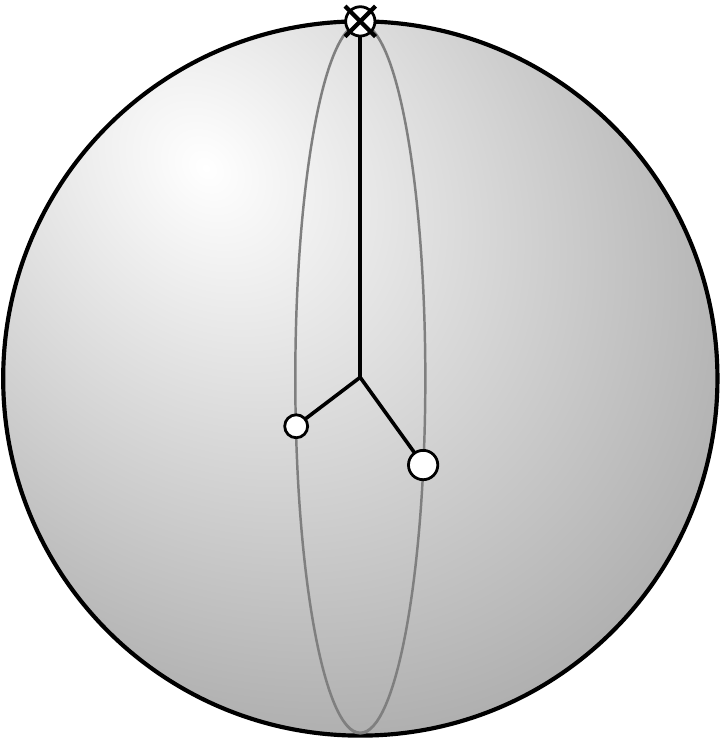} \put(-14,0){(c)}
        \end{overpic}
        \begin{overpic}[scale=.22]{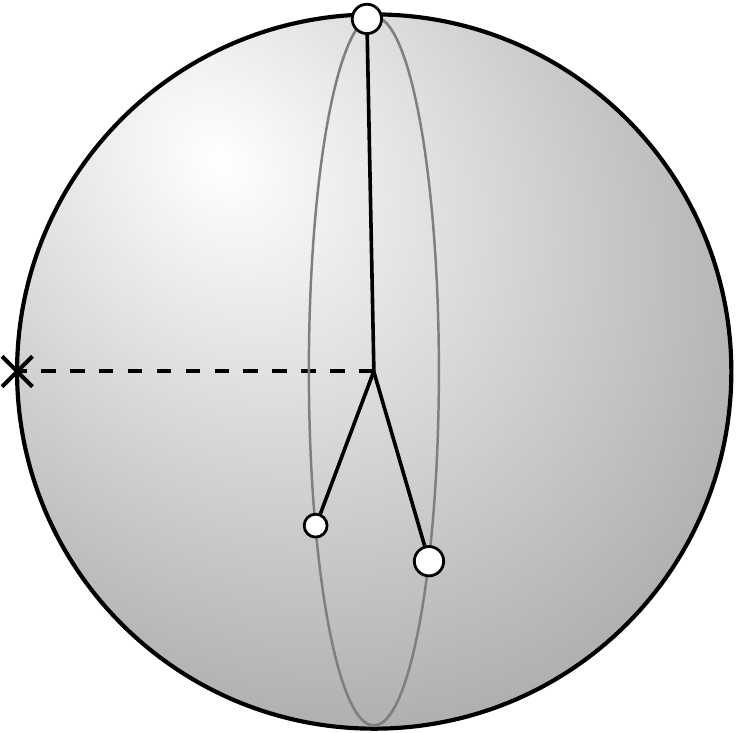} \put(-14,0){(d)}
        \end{overpic}
        \begin{overpic}[scale=.22]{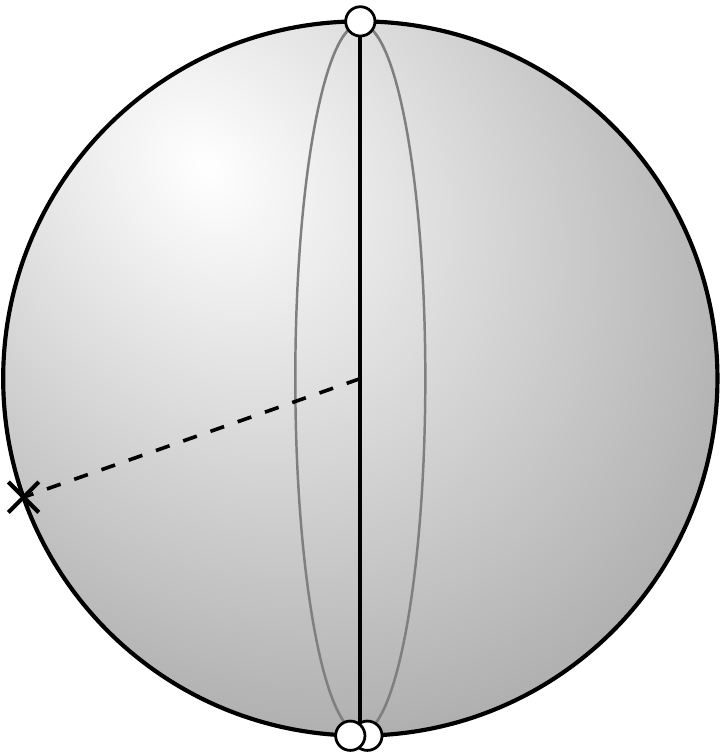} \put(-14,0){(e)}
        \end{overpic}
      \end{center}
    \end{minipage}
    \begin{minipage}{86mm}
      \vspace{4mm}
      \begin{center}
        \begin{overpic}[scale=.66]{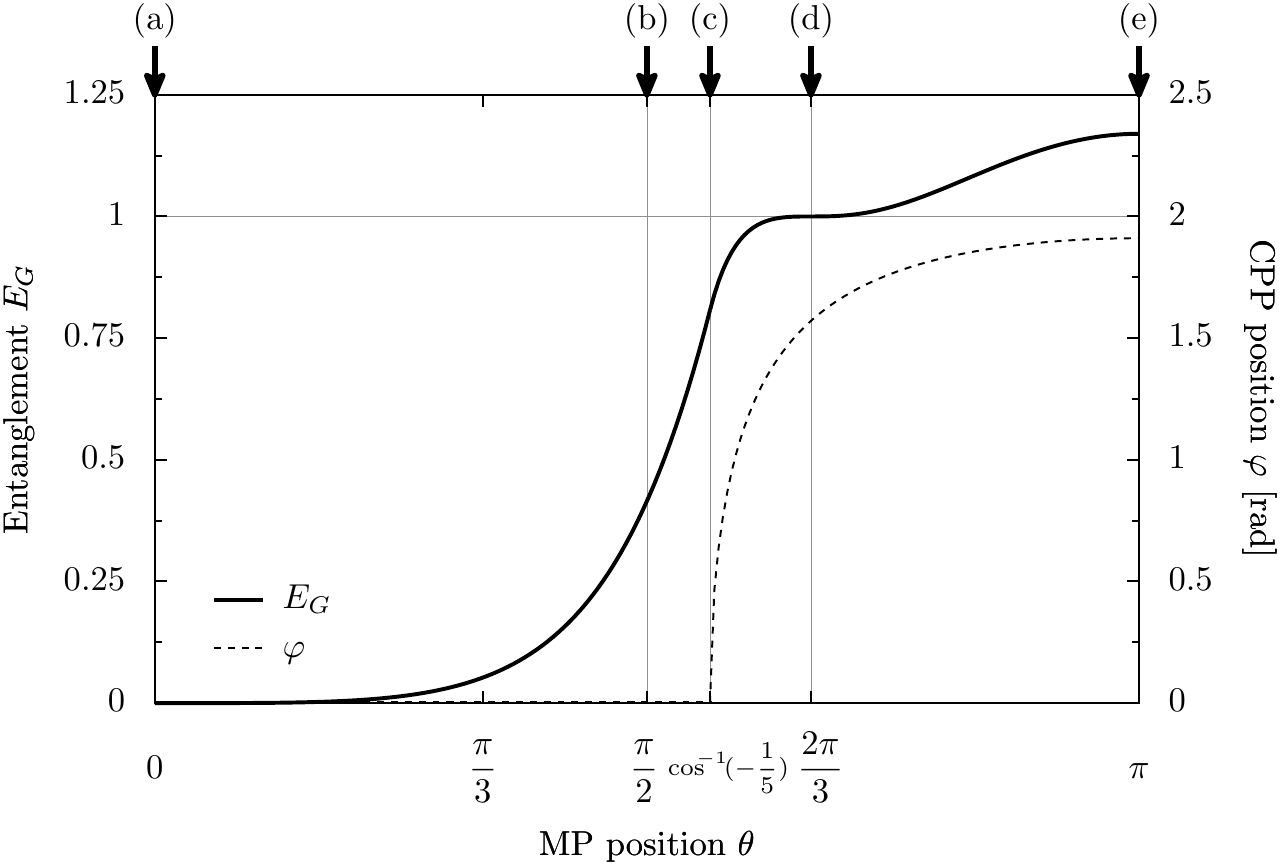}
        \end{overpic}
      \end{center}
    \end{minipage}
  \end{center}
  \caption{\label{3_graph} Change of the entanglement and the location
    of the CPP when the MP distribution is modified. The position of
    the CPP does not change until the two moving MPs have reached a
    latitude slightly below the equator. From the distribution (c)
    onwards, the CPP rapidly moves southwards and reaches the equator
    at the GHZ state (d). After that, the location of the CPP and the
    entanglement changes only weakly until the W state (e) is
    reached.}
\end{figure}

It is insightful to examine how the CPPs and the entanglement change
when the MP distribution of the underlying state is modified.
Starting out with three MPs lying on the north pole, two of the MPs
are moved southwards on opposite sides (cf. \fig{3_graph}), describing
an isosceles triangle with the remaining MP on the north pole.  Using
the abbreviations $\co_{\theta} = \cos (\theta / 2)$ and $\si_{\theta}
= \sin (\theta / 2)$, the MPs have the form
\begin{equation}\label{mp_3_form}
  \begin{split}
    | \phi_{1} \rangle & = | 0 \rangle \enspace , \\
    | \phi_{2,3} \rangle & = \co_{\theta} | 0 \rangle \pm \I \,
    \si_{\theta} | 1 \rangle \enspace ,
  \end{split}
\end{equation}
with the parametrization $\theta \in [0,\pi]$. The form of the
underlying quantum state follows from \eq{majorana_definition}
as
\begin{equation}\label{3_mp_state1}
  \vert \psi \rangle  = \frac{3 \, \co_{\theta}^2 \, \vert 000 \rangle +
    \si_{\theta}^2 \left( \vert 011 \rangle \! + \! \vert 101 \rangle
      \! + \! \vert 110 \rangle \right) }
  {\sqrt{9 \, \co_{\theta}^4 + 3 \, \si_{\theta}^4 }} \enspace .
\end{equation}
This state is positive, so Lemma \ref{lem_positive_symmetric} asserts
the existence of at least one positive CPP.  With the ansatz $\vert
\sigma \rangle = \co_{\varphi} \vert 0 \rangle + \si_{\varphi} \vert 1
\rangle$ for the CPP, the position of the CPP is found by calculating
the absolute maximum of $\vert \langle \psi \vert \sigma
\rangle^{\otimes 3} \vert$.  From this it is found that the parameter
$\varphi (\theta)$ of the CPP depends on the parameter $\theta$ of the
MPs as follows:
\begin{equation}\label{3_mp_state2}
  \co_{\varphi}^2 = \si_{\theta}^2 / (6 \si_{\theta}^2 - 3) \enspace .
\end{equation}
The permitted values of the left-hand side are [0,1], but the
right-hand side lies outside this range for $\theta < \pi - \arccos
(1/5)$. For these values the CPP is fixed at $\vert \sigma \rangle =
\vert 0 \rangle$.  \Fig{3_graph} shows how the CPP parameter $\varphi
( \theta )$ changes with $\theta$.  It is seen that from $\theta = \pi
- \arccos (1/5)$ onwards, the CPP abruptly leaves the north pole and
moves towards the south pole along the prime meridian.  From Equations
\eqsimple{3_mp_state1} and \eqsimple{3_mp_state2} the amount of
entanglement is easily calculated and is displayed in \fig{3_graph}.
$E_{\text{G}}$ is monotonously increasing \cite{Tamaryan09} and
reaches a saddle point at the GHZ state ($\theta = 2 \pi / 3$).

\subsection{Entanglement and extremal point distributions}
\label{extremal_point}

The main point of interest in this paper is the study of maximally
entangled symmetric states. For this the Majorana representation is
extremely helpful, because it allows the optimization problem of
maximizing the entanglement to be written in a simple form.  With the
help of \eq{bloch_product}, the min-max problem \eqsimple{minmax} for
finding the maximally entangled state can be reformulated as
\begin{equation}\label{maj_problem}
  \min_{ \{ | \phi_i \rangle \}} \frac{1}{\sqrt{K}}
  \left( \max_{ | \sigma \rangle } \, \prod_{i=1}^{n} \,
    | \langle \sigma | \phi_i \rangle | \right) \enspace .
\end{equation}
This `Majorana problem' bears all the properties of an optimization
problem on the surface of a sphere in $\mathbb{R}^3$.  These kinds of
problems deal with arrangements of a finite number of points on a
sphere so that an extremal property is fulfilled \cite{Whyte52}. Two
well-known members, T\'{o}th's problem and Thomson's problem, have
been extensively studied in the past.

\textbf{T\'{o}th's problem,} also known as Fejes' problem and Tammes'
problem, asks how $n$ points have to be distributed on the unit sphere
so that the minimum distance of all pairs of points becomes maximal
\cite{Whyte52}. This problem was first raised by the biologist Tammes
in 1930 when trying to explain the observed distribution of pores on
pollen grains \cite{Tammes30}. Recasting the $n$ points as unit
vectors $\mathbf{r}_{i} \in \mathbb{R}^3$, the following cost function
needs to be maximized:
\begin{equation}
   f_{\text{T\'{o}th}} ( \mathbf{r}_1 , \mathbf{r}_2 , \dots ,
   \mathbf{r}_{n} ) =
   \min_{i < j} \, | \mathbf{r}_{i} - \mathbf{r}_{j} | \enspace .
\end{equation}
The point configuration that solves this problem is called a spherical
code or sphere packing \cite{Weisstein}. The latter term refers to the
equivalent problem of placing $n$ identical spheres of maximal
possible radius around a central unit sphere, touching the unit sphere
at the points that solve T\'{o}th's problem.

\textbf{Thomson's problem,} also known as the Coulomb problem, asks
how $n$ point charges can be distributed on the surface of a sphere so
that the potential energy is minimized.  The charges interact with
each other only through Coulomb's inverse square law. Devised by
J. J. Thomson in 1904, this problem raises the question about the
stable patterns of up to 100 electrons on a spherical surface
\cite{Thomson04}.  Its cost function is given by the Coulomb energy
and needs to be minimized.
\begin{equation}
    f_{\text{Thomson}} ( \mathbf{r}_1 , \mathbf{r}_2 , \dots ,
    \mathbf{r}_{n} ) = \sum_{i < j} \, | \mathbf{r}_{i} -
    \mathbf{r}_{j} |^{-1} \enspace .
\end{equation}
The original motivation for Thomson's problem was to determine the
stable electron distribution of atoms in the plum pudding model.
While this model has been superseded by modern quantum theory, there
is a wide array of novel applications for Thomson's problem or its
generalization to other interaction potentials.  Among these are
multi-electron bubbles in liquid $^4$He \cite{Leiderer95}, surface
ordering of liquid metal drops confined in Paul traps \cite{Davis97},
the shell structure of spherical viruses \cite{Marzec93},
`colloidosomes' for encapsulating biochemically active substances
\cite{Dinsmore02}, fullerene patterns of carbon atoms \cite{Kroto85}
and the Abrikosov lattice of vortices in superconducting metal shells
\cite{Dodgson97}.

Exact solutions to T\'{o}th's problem are only known for
$n_{\text{To}} = 2-12,24$ points \cite{Erber91}, and in Thomson's
problem for $n_{\text{Th}} = 2-8,12$ points \cite{Erber91,Whyte52}.
Despite the different definitions of the two problems, they share the
same solutions for $n = 2-6,12$ points \cite{Leech57}.  Numerical
solutions are, furthermore, known for a wide range of $n$ in both
problems \cite{Ashby86,Altschuler94,Sloane,Wales}.

The solutions to $n = 2, 3$ are trivial and given by the dipole and
equilateral triangle, respectively.  For $n = 4,6,8,12,20$ the
Platonic solids are natural candidates, but they are the actual
solutions only for $n = 4,6,12$ \cite{Berezin85}. For $n = 8,20$ the
solutions are not Platonic solids and are different for the two
problems. We will cover the solutions for $n=4-12$ in more detail
alongside the Majorana problem in
\sect{maximally_entangled_symmetric_states}.

On symmetry grounds, one could expect that the center of mass of the
$n$ points always coincides with the sphere's middle point. This is,
however, not the case, as the solution to T\'{o}th's problem for $n=7$
\cite{Erber91} or the solution to Thomson's problem for $n = 11$ shows
\cite{Erber91, Ashby86}. Furthermore, the solutions need not be
unique. For T\'{o}th's problem, the first incident of this is $n = 5$
\cite{Ogilvy51}, and for Thomson's problem at $n=15$ \cite{Erber91}
and $n = 16$ \cite{Ashby86}.  These aspects show that it is, in
general, hard to make statements about the form of the `most spread
out' point distributions on the sphere.  The Majorana problem
\eqsimple{maj_problem} is considered to be equally tricky,
particularly with the normalization factor $K$ depending on the MPs.
Furthermore, the MPs of the solution need not all be spread out far
from each other, as demonstrated by the three qubit $| \text{W}
\rangle$ state with its two coinciding MPs.

\section{States and symmetries of MP and CPP distributions}
\label{analytic}

In this section, results for the interdependence between the form of
$n$ qubit symmetric states and their Majorana representation will be
derived. More specifically, it will be examined what the distributions
of MPs and CPPs look like for states whose coefficients are real,
positive or vanishing. In some of these cases the MPs or CPPs have
distinct patterns on the sphere, which can be described by symmetries.
In this context, care has to be taken as to the meaning of the word
`symmetric'. Permutation-\emph{symmetric} states were introduced in
\sect{positive_and_symm}, and only these states can be represented by
point distributions on the Majorana sphere.  For some of these
symmetric states, their MP distribution exhibits symmetry properties
on the sphere.  Examples of this can be found in \fig{ghz_w_pic},
where the GHZ state and W state have \emph{rotational symmetries}
around the Z-axis, as well as \emph{reflective symmetries} along some
planes.

Let $| \psi \rangle_{\text{s}} = \sum_{k = 0}^{n} a_k | S_{k} \rangle$
be a general symmetric state of $n$ qubits.  To understand the
relationship between the state's coefficients and the Majorana
representation, consider the effect of symmetric LUs.  A symmetric LU
acting on the Hilbert space of an $n$ qubit system is defined as the
$n$-fold tensor product of a single-qubit unitary operation:
$U^{\text{s}} = U \otimes \cdots \otimes U$. This is precisely the LU
map that was shown in \eq{m_def_1} and \eqsimple{lusphere} to map
every symmetric state to another symmetric state.

Considering the Hilbert space of a single qubit, the rotation operator
for $Z$-axis rotations of the qubit is
\begin{equation}\label{z_rotationmatrix}
  R_z (\theta) =
  \begin{pmatrix}
    1 & 0 \\
    0 & \E^{\I \theta}
  \end{pmatrix} \enspace ,
\end{equation}
and the rotation operator for $Y$-axis rotations of the qubit is
\begin{equation}\label{y_rotationmatrix}
  R_y (\theta) =
  \begin{pmatrix}
    \cos \frac{\theta}{2} & - \sin \frac{\theta}{2}  \vspace*{0.8mm} \\
    \sin \frac{\theta}{2} & \phantom{-} \cos \frac{\theta}{2}
  \end{pmatrix} \enspace .
\end{equation}
$R_z$ changes the relative phase, but not the absolute value of the
qubit's coefficients.  Conversely, $R_y$ changes the absolute value,
but not the relative phase of the coefficients.  From
\eq{majorana_definition} it is easily seen that $R_z$ and $R_y$ pass
this behavior on to the symmetric LUs $R^{\, \text{s}}_z :=
R_z^{\otimes n}$ and $R^{\, \text{s}}_y := R_y^{\otimes n}$. For
example, the effect of $R^{\, \text{s}}_z$ on $| \psi
\rangle_{\text{s}}$ is
\begin{equation}\label{rot_z}
  R^{\, \text{s}}_z (\theta) \, | \psi \rangle_{\text{s}} =
  \sum_{k = 0}^{n} a_k \E^{ \I k \theta} \, | S_{k} \rangle \enspace .
\end{equation}
From this it is easy to determine the conditions for the MPs of $|
\psi \rangle_{\text{s}}$ having a rotational symmetry around the
Z-axis, i.e. $R^{\, \text{s}}_z (\theta) \, | \psi \rangle_{\text{s}}
= | \psi \rangle_{\text{s}}$ (up to a global phase) for some $\theta <
2 \pi$.  From \eq{rot_z} it is clear that the possible rotational
angles (up to multiples) are restricted to $\theta = 2 \pi / m$, with
$m \in \mathbb{N}$, $1<m \leq n$.  The necessary and sufficient
conditions are:

\begin{lem}\label{rot_symm}
  The MP distribution of a symmetric $n$ qubit state $| \psi
  \rangle_{\text{s}}$ is rotationally symmetric around the $Z$-axis
  with rotational angle $\theta = 2 \pi / m$ ($\, 1<m \leq n$) iff
\begin{equation}\label{rot_cond}
  \forall \{ k_i , k_j | \, a_{k_{i}} \neq 0 \wedge
  a_{k_{j}} \neq 0 \} : ( k_i - k_j ) \bmod m = 0
\end{equation}
\end{lem}
\begin{proof}
  \Eq{rot_cond} is equivalent to: $\exists \: l \in \mathbb{Z} :
  \forall \{ k | \, a_{k} \neq 0 \} : k \bmod m = l$.  From this it
  follows that $R^{\, \text{s}}_z (2 \pi / m) \, | \psi
  \rangle_{\text{s}} = \sum_{k} a_k \exp(\I 2 \pi k/m) \vert S_k
  \rangle = \sum_{k} a_k \exp(\I 2 \pi l/m) \vert S_k \rangle = \E^{\I
    \delta} | \psi \rangle_{\text{s}}$, with $\delta = 2 \pi l/m$.
\end{proof}

In other words, a sufficient number of coefficients need to vanish,
and the spacings between the remaining coefficients must be multiples
of $m$. For example, a symmetric state of the form $\vert \psi
\rangle_{\text{s}} = a_3 \vert S_3 \rangle + a_7 \vert S_7 \rangle +
a_{15} \vert S_{15} \rangle$ is rotationally symmetric with $\theta =
\pi / 2$, because the spacings between non-vanishing coefficients are
multiples of $4$.

\begin{lem}\label{maj_max}
  Every maximally entangled symmetric state $| \Psi
  \rangle_{\text{s}}$ of $n$ qubits has at least two different CPPs.
\end{lem}
\begin{proof}
  The cases $n = 2, 3$ are trivial, because their maximally entangled
  states (Bell states and W state, respectively) have an infinite
  number of CPPs.  For $n > 3$, we consider a symmetric state $| \psi
  \rangle$ with only one CPP $| \sigma \rangle$ and show that $| \psi
  \rangle$ cannot be maximally entangled.

  Because of the LU invariance on the Majorana sphere
  (cf. \eq{lusphere}), we can take the CPP to be the north pole,
  i.e. $| \sigma \rangle = | 0 \rangle$.  Denoting a single qubit with
  $| \omega \rangle = \co_{\theta} |0\rangle + \E^{\I \varphi}
  \si_{\theta} |1\rangle$, the smooth and continuous overlap function
  $g( |\omega \rangle ) = | \langle \psi | \omega \rangle^{\otimes
    n}|$ then has its absolute maximum at $| \omega \rangle = | 0
  \rangle$.  For any other local maximum $| \omega ' \rangle$, the
  value of $g(| \omega ' \rangle)$ is smaller than $g(|0\rangle)$ and
  therefore an infinitesimal change in the MPs of $| \psi \rangle$
  cannot lead to a CPP outside a small neighborhood of the north pole.

  We will now present an explicit variation of $| \psi \rangle$ that
  increases the entanglement.  $| \psi \rangle = \sum_{k=0}^n a_k |
  S_k \rangle$ has complex coefficients that fulfil $\langle \psi |
  \psi \rangle= 1$, as well as $a_0 > 0$ and $a_1 = 0$ \footnote{$a_0$
    can be set positive by means of the global phase, and for
    symmetric states with $| 0 {\rangle}$ as a CPP it is easy to
    verify that $a_1 = 0$ is a necessary condition for the partial
    derivatives of $| {\langle} \psi | \omega {\rangle}^{\otimes n}|$
    being zero at $| \omega {\rangle} = | 0 {\rangle}$.}.  Define the
  variation as $| \psi_{\epsilon} \rangle = (a_0 - \epsilon ) | S_0
  \rangle + \sum_{k=2}^{n-1} a_k | S_k \rangle + (a_n + \epsilon a_0/
  a_n^{*} ) | S_n \rangle$, with $\epsilon \ll 1$.  This state fulfils
  the requirement $| \psi_{\epsilon} \rangle \stackrel{\epsilon
    \rightarrow 0}{\longrightarrow} | \psi \rangle$, and is
  normalized: $\langle \psi_{\epsilon} | \psi_{\epsilon} \rangle = 1 +
  \Order (\epsilon^2)$.  We now investigate the values of
  $g_{\epsilon} ( |\omega \rangle ) = |\langle \psi_{\epsilon} |
  \omega \rangle^{\otimes n}|$ around the north pole. In this area $|
  \omega \rangle = (1 - (\theta^2 / 8) + \Order (\theta^4))
  |0\rangle + \E^{\I \varphi} ( (\theta / 2) - \Order (\theta^3))
  |1\rangle$, hence $g_{\epsilon} ( |\omega \rangle ) = | (1 -
  (\theta^2 / 8))^n (a_0 - \epsilon ) | + \Order (\epsilon^2,
  \theta^2) = | a_0 - \epsilon | + \Order (\epsilon^2, \theta^2) <
  |a_0| = g ( |0 \rangle )$ for small, but nonzero $\epsilon$ and
  $\theta$.  Therefore the absolute maximum of $g_{\epsilon}$ is
  smaller than that of $g$, and $| \psi_{\epsilon} \rangle$ is more
  entangled than $| \psi \rangle$.
\end{proof}

\subsection{Real symmetric states}
\label{maximally_entangled_real_states}

For symmetric states with real coefficients, the following lemma
asserts a reflection symmetry of the MPs and CPPs with respect to the
$X$-$Z$-plane that cuts the Majorana sphere in half.  In mathematical
terms, the MPs and CPPs exhibit a reflection symmetry with respect to
the $X$-$Z$-plane iff for each MP $| \phi_i \rangle = \co_{\theta}
|0\rangle + \E^{\I \varphi} \si_{\theta} |1\rangle$ the complex
conjugate point $| \phi_i \rangle^{*} = \co_{\theta} |0\rangle + \E^{-
  \I \varphi} \si_{\theta} |1\rangle$ is also a MP, and the same holds
for CPPs too.

\begin{lem}\label{maj_real}
  Let $| \psi \rangle_{\text{s}}$ be a symmetric state of $n$ qubits.
  $| \psi \rangle_{\text{s}}$ is real iff all its MPs are reflective
  symmetric with respect to the $X$-$Z$-plane of the Majorana sphere.
\end{lem}
\begin{proof}
  ($\Rightarrow$) Let $| \psi \rangle_{\text{s}}$ be a real state.
  Then $| \psi \rangle_{\text{s}} = | \psi \rangle_{\text{s}}^{*}$,
  and since Majorana representations are unique, $| \psi
  \rangle_{\text{s}}$ has the same MPs as $| \psi
  \rangle_{\text{s}}^{*}$.  Therefore the complex conjugate $| \phi_i
  \rangle^{*}$ of each non-real MP $| \phi_i \rangle$ is also a MP.

  ($\Leftarrow$) Let the MPs of $| \psi \rangle_{\text{s}}$ be
  symmetric with respect to the $X$-$Z$-plane. Then for every nonreal
  MP $| \phi_i \rangle$ its complex conjugate $| \phi_i \rangle^{*}$
  is also a MP.  Because $\left( | \phi_i \rangle | \phi_i \rangle^{*}
    + | \phi_i \rangle^{*} | \phi_i \rangle \right)$ is real, it
  becomes clear, from the permutation over all MPs in
  \eq{majorana_definition}, that the overall state $| \psi
  \rangle_{\text{s}}$ is real, too.
\end{proof}

The reflective symmetry of the MPs naturally leads to the same
symmetry for the CPPs.

\begin{cor}\label{cpp_real}
  Let $| \psi \rangle_{\text{s}}$ be a symmetric state of $n$
  qubits. If $| \psi \rangle_{\text{s}}$ is real, then all its CPPs
  are reflective symmetric with respect to the $X$-$Z$-plane of the
  Majorana sphere.
\end{cor}
\begin{proof}
  Lemma \ref{maj_real} asserts that for every MP $| \phi_i \rangle$ of
  $| \psi \rangle_{\text{s}}$, the complex conjugate $| \phi_i
  \rangle^{*}$ is also a MP.  By considering the complex conjugate of
  the optimization problem \eqsimple{maj_problem}, it becomes clear
  that for any CPP $| \sigma \rangle$ the complex conjugate $| \sigma
  \rangle^{*}$ is also a CPP.
\end{proof}

\subsection{Positive symmetric states}
\label{maximally_entangled_positive_states}

For symmetric states with positive coefficients, strong results can be
obtained with regard to the number and locations of the CPPs.  In
particular, for non-Dicke states it is shown that there are at most
$2n-4$ CPPs and that non-positive CPPs can only exist if the MP
distribution has a rotational symmetry around the Z-axis.
Furthermore, the CPPs can only lie at specified azimuthal angles on
the sphere, namely those that are `projected' from the meridian of
positive Bloch vectors by means of the Z-axis rotational symmetry
(see, e.g., the positive seven qubit state shown in \fig{bloch_7}).

Dicke states constitute a special case due to their continuous
azimuthal symmetry. The two Dicke states $| S_0 \rangle$ and $| S_n
\rangle$ are product states, with all their MPs and CPPs lying on the
north and the south pole, respectively. For any other Dicke state $| S_k
\rangle$ the MPs are shared between the two poles, and the CPPs form a
continuous horizontal ring with inclination $\theta = 2 \arccos
\sqrt{{n-k}/{n}}$.

\begin{lem}\label{cpp_mer}
  Let $| \psi \rangle_{\text{s}}$ be a positive symmetric state of $n$
  qubits, excluding the Dicke states.

  \begin{enumerate}
  \item[(a)] If $| \psi \rangle_{\text{s}}$ is rotationally symmetric
    around the Z-axis with minimal rotational angle $2 \pi / m$, then
    all its CPPs $| \sigma (\theta, \varphi) \rangle = \co_{\theta}
    |0\rangle + \E^{\I \varphi} \si_{\theta} |1\rangle$ are restricted
    to the $m$ azimuthal angles given by $\varphi = \varphi_{r} = 2
    \pi r / m$ with $r \in \mathbb{Z}$.  Furthermore, if $| \sigma
    (\theta, \varphi_{r} ) \rangle$ is a CPP for some $r$, then it is
    also a CPP for all other values of $r$.

  \item[(b)] If $| \psi \rangle_{\text{s}}$ is not rotationally
    symmetric around the Z-axis, then all its CPPs are positive.
  \end{enumerate}
\end{lem}
\begin{proof}
  The proof runs similar to the one of Lemma \ref{lem_positive}, where
  the existence of at least one positive CPP is established.  We use
  the notations $| \psi \rangle_{\text{s}} = \sum_{k} a_{k} | S_{k}
  \rangle$ with $a_k \geq 0$, and $| \lambda \rangle = | \sigma
  \rangle^{\otimes n}$.

  Case (a): Consider a non-positive CPP $| \sigma \rangle =
  \co_{\theta} |0\rangle + \E^{\I \varphi} \si_{\theta} |1\rangle$
  with $\varphi = 2 \pi r / m$ and $r \in \mathbb{R}$, and define $|
  \widetilde{\sigma} \rangle = \co_{\theta} |0\rangle + \si_{\theta}
  |1\rangle$. Then $|\langle \lambda | \psi \rangle_{\text{s}} | = |
  \sum_k \E^{\I k \varphi} a_k \co_{\theta}^{n-k} \si_{\theta}^{k}
  \sqrt{n/k} | \leq \sum_k a_k \co_{\theta}^{n-k} \si_{\theta}^{k}
  \sqrt{n/k} = |\langle \widetilde{\lambda} | \psi \rangle_{\text{s}}
  |$. If this inequality is strict, then $| \sigma \rangle$ is not a
  CPP.  This would be a contradiction, so it must be an equality.
  Thus, for any two indices $k_i$ and $k_j$ of non-vanishing
  coefficients $a_{k_i}$ and $a_{k_j}$, the following must hold:
  $\E^{\I k_{i} \varphi} = \E^{\I k_{j} \varphi}$.  This can be
  reformulated as $k_{i} r \bmod m = k_{j} r$, or equivalently
  \begin{equation}\label{pos_mer_cond}
    ( k_{i} - k_{j} ) \, r \bmod m = 0 \enspace .
  \end{equation}
  Because $\varphi = 2 \pi / m$ is the minimal rotational angle, $m$
  is the largest integer that satisfies \eq{rot_cond}, and thus there
  exist $k_i$ and $k_j$ with $a_{k_i}, a_{k_j} \neq 0$ s.t. $k_i - k_j
  = m$.  From this and from \eq{pos_mer_cond}, it follows that $r \in
  \mathbb{Z}$.  Therefore $| \sigma (\theta, \varphi_{r} ) \rangle$ is
  a CPP if and only if $r$ is an integer.

  Case (b): Considering a CPP $| \sigma \rangle = \co_{\theta}
  |0\rangle + \E^{\I \rho} \si_{\theta} |1\rangle$ with $\rho = 2 \pi
  r$ and $r \in \mathbb{R}$, we need to show that $r \in \mathbb{Z}$.
  Defining $| \widetilde{\sigma} \rangle = \co_{\theta} |0\rangle +
  \si_{\theta} |1\rangle$, and using the same line of argumentation as
  above, the equation $\E^{\I k_{i} \rho} = \E^{\I k_{j} \rho}$ must
  hold for any pair of non-vanishing $a_{k_i}$ and $a_{k_j}$.  This is
  equivalent to
  \begin{equation}\label{pos_mer_cond2}
    ( k_{i} - k_{j} ) \, r \bmod 1 = 0 \enspace ,
  \end{equation}
  or $( k_{i} - k_{j} ) \, r \in \mathbb{Z}$, in particular $r \in
  \mathbb{Q}$. If there exist indices $k_{i}$ and $k_{j}$ of
  non-vanishing coefficients s.t. $k_{i} - k_{j} = 1$, then $r \in
  \mathbb{Z}$, as desired.  Otherwise consider $r = x/y$ with coprime
  $x,y \in \mathbb{N}$, $x < y$.  Because $| \psi \rangle_{\text{s}}$
  is not rotationally symmetric, the negation of Lemma \ref{rot_symm}
  yields that, for any two $k_{i}$ and $k_{j}$ ($a_{k_{i}}, a_{k_{j}}
  \neq 0$) with $k_{i} - k_{j} = \alpha > 1$, there must exist a
  different pair $k_{p}$ and $k_{q}$ ($a_{k_p}, a_{k_q} \neq 0$) with
  $k_{p} - k_{q} = \beta > 1$ s.t.  $\alpha$ is not a multiple of
  $\beta$ and vice versa.  From $r = x/y$ and \eq{pos_mer_cond2}, it
  follows that $y = \alpha$ as well as $y = \beta$. This is a
  contradiction, so $r \in \mathbb{Z}$.
\end{proof}

With this result about the confinement of the CPPs to certain
azimuthal angles, it is possible to derive upper bounds on the number
of CPPs.

\begin{theo}\label{maj_max_pos_zero}
  The Majorana representation of every positive symmetric state $|
  \psi \rangle_{\text{s}}$ of $n$ qubits, excluding the Dicke states,
  belongs to one of the following three mutually exclusive classes.

  \begin{enumerate}
  \item[(a)] $| \psi \rangle_{\text{s}}$ is rotationally symmetric
    around the Z-axis, with only the two poles as possible CPPs.

  \item[(b)] $| \psi \rangle_{\text{s}}$ is rotationally symmetric
    around the Z-axis, with at least one CPP being non-positive.

  \item[(c)] $| \psi \rangle_{\text{s}}$ is not rotationally symmetric
    around the Z-axis, and all CPPs are positive.
  \end{enumerate}
  Regarding the CPPs of states from class (b) and (c), the following
  assertions can be made for $n \geq 3$:
  \begin{enumerate}
  \item[(b)] If both poles are occupied by at least one MP each, then
    there are at most $2n-4$ CPPs, else there are at most $n$ CPPs.

  \item[(c)] There are at most $\lceil \tfra{n+2}{2} \rceil$ CPPs
    \footnote{The ceiling function ${\lceil} x {\rceil}$ is the
      smallest integer not less than $x$.}.
  \end{enumerate}
\end{theo}
\begin{proof}
  Starting with the first part of the theorem, case (c) has already
  been shown in Lemma \ref{cpp_mer}, so consider states $| \psi
  \rangle_{\text{s}}$ with a rotational symmetry $\varphi = 2 \pi / m$
  around the Z-axis.  If all CPPs are either $\vert 0 \rangle$ or
  $\vert 1 \rangle$, then we have case (a), otherwise there is at
  least one CPP $| \sigma \rangle$ which does not lie on a pole.  If
  this $| \sigma \rangle$ is non-positive, then we have case (b), and
  if $| \sigma \rangle$ is positive, then Lemma \ref{cpp_mer} states
  the existence of another, non-positive CPP, thus again resulting in
  case (b).

  The proof of the second part of the theorem is a bit involved and
  can be found in \app{max_cpp_number}.
\end{proof}

\section{Numerical solutions for up to twelve qubits}
\label{maximally_entangled_symmetric_states}

In this section, we present the numerically determined solutions of
the Majorana problem for up to 12 qubits.  In order to find these, the
results from the previous sections were extremely helpful.  This is
because an exhaustive search over the set of all symmetric states
quickly becomes unfeasible, even for low numbers of qubits and because
the min-max problem \eqsimple{maj_problem} is too complex to allow for
straightforward analytic solutions.

Among the results particularly helpful for our search are Lemma
\ref{cpp_mer} and Theorem \ref{maj_max_pos_zero}. For positive states
they strongly restrict the possible locations of CPPs, and thus
greatly simplify the calculation of the entanglement.  It then
suffices to determine only the positive CPPs because all other CPPs
automatically follow from the Z-axis symmetry (if any exists).  We
will see that this result is especially powerful for the Platonic
solids in the cases $n=4,6$ where the location of the CPPs can be
immediately determined from this argument alone, without the need to
do any calculations.

From the definition \eqsimple{geo_meas} of $E_{\text{G}}$, it is clear
that the exact amount of entanglement of a given state (or its
corresponding MP distribution) is automatically known once the
location of at least one CPP is known. A numerical search over the set
of positive symmetric states often detects the maximally entangled
state to be of a particularly simple form, enabling us to express the
exact positions of its MPs and CPPs analytically. In some cases,
however, no analytical expressions were found for the positions of the
CPPs and/or MPs. In these cases the exact amount of entanglement
remains unknown, although it can be numerically approximated with high
precision.

In this way we can be quite confident of finding the maximally
entangled positive symmetric state.  In the general symmetric case we
do not have as many tools, so the search is over a far bigger set of
possible states, and we can be less confident in our results. We
therefore focus our search to sets of states that promised high
entanglement. Such states include those with highly spread out MP
distributions, particularly the solutions of the classical
optimization problems.  From these results we can propose candidates
for maximal symmetric entanglement.

For two and three qubits, the maximally entangled symmetric states are
known and were discussed in \sect{majorana_representation}, so we
start with four qubits.  A table summarizing the amount of
entanglement of the numerically determined maximally entangled
positive symmetric as well as the entanglement of the candidates for
the general symmetric case can be found in \app{entanglement_table}.

\subsection{Four qubits}\label{majorana_four}

For four points, both T\'{o}th's and Thomson's problem are solved by
the vertices of the regular tetrahedron \cite{Whyte52}. The numerical
search for the maximally entangled symmetric state returns the
Platonic solid too.  Recast as MPs, the vertices are
\begin{equation}\label{4_mp}
  \begin{split}
    | \phi_{1} \rangle & = | 0 \rangle \enspace , \\
    | \phi_{2} \rangle & = \tfra{1}{\sqrt{3}} | 0 \rangle +
    \sqrt{\tfra{2}{3}} | 1 \rangle \enspace , \\
    | \phi_{3} \rangle & = \tfra{1}{\sqrt{3}} | 0 \rangle +
    \E^{\I 2 \pi / 3} \sqrt{\tfra{2}{3}} | 1 \rangle \enspace , \\
    | \phi_{4} \rangle & = \tfra{1}{\sqrt{3}} | 0 \rangle +
    \E^{\I 4 \pi / 3} \sqrt{\tfra{2}{3}} | 1 \rangle \enspace .
  \end{split}
\end{equation}
The symmetric state constructed from these MPs by means of
\eq{majorana_definition} shall be referred to as the `tetrahedron
state'.  Its form is $| \Psi_{4} \rangle = \sqrt{1/3} \, | S_{0}
\rangle + \sqrt{2/3} \, | S_{3} \rangle$, and its MP distribution is
shown in \fig{bloch_4}.  Because the state is positive and has a
rotational symmetry around the Z-axis, Lemma \ref{cpp_mer} restricts
the possible CPP locations to the three half-circles $| \sigma
(\theta, \varphi ) \rangle$ with $\varphi = 0, 2 \pi / 3, 4 \pi / 3$.
\begin{figure}[b]
  \begin{center}
    \begin{overpic}[scale=.5]{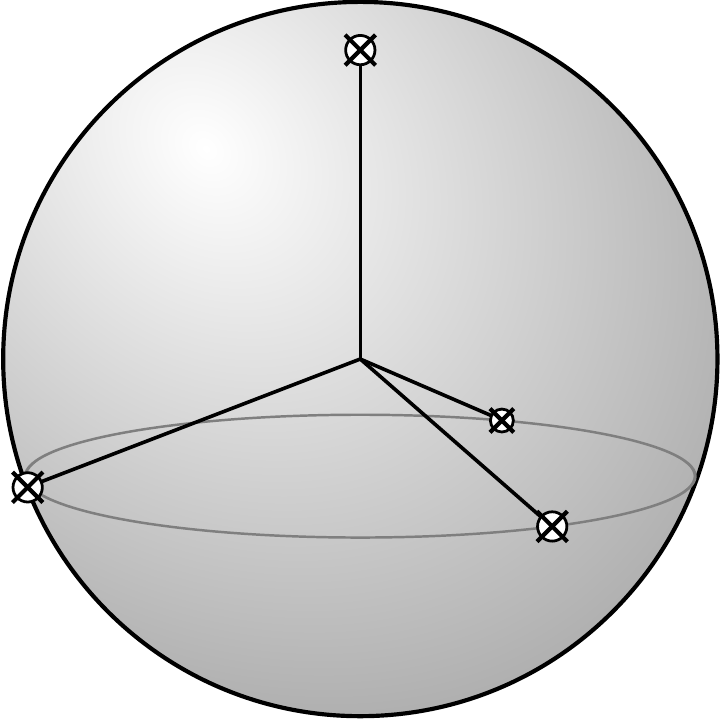}
      \put(30,86){$|  \phi_1 \rangle$}
      \put(-11,21){$| \phi_2 \rangle$}
      \put(66,16){$|  \phi_3 \rangle$}
      \put(68,48){$|  \phi_4 \rangle$}
    \end{overpic}
  \end{center}
  \caption{\label{bloch_4} MPs and CPPs of the `tetrahedron state'.}
\end{figure}
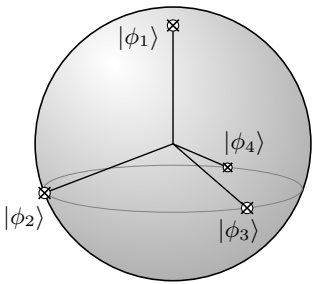
From the symmetry of the Platonic solid it is clear that the MP
distribution of \fig{bloch_4} can be rotated s.t. $| \phi_{2}
\rangle$, $| \phi_{3} \rangle$ or $| \phi_{4} \rangle$ is moved to the
north pole, with the actual distribution (and thus $| \Psi_{4}
\rangle$) remaining unchanged.  Each of these rotations, however,
gives rise to new restrictions on the location of the CPPs mediated by
Lemma \ref{cpp_mer}. The intersections of all these restrictions are
the four points where the MPs lie.  This yields the result that $|
\Psi_{4} \rangle$ has four CPPs, with their Bloch vectors being the
same as those in \eq{4_mp}.  From this the amount of entanglement
follows as $E_{\text{G}} ( | \Psi_{4} \rangle ) = \log_2 3 \approx
1.59$.

\subsection{Five qubits}\label{majorana_five}

For five points, the solution to Thomson's problem is given by three
of the charges lying on the vertices of an equatorial triangle and the
other two lying at the poles \cite{Ashby86,Marx70}.  This is also a
solution to T\'{o}th's problem, but it is not unique
\cite{Ogilvy51,Schutte51}.  The corresponding quantum state, the
`trigonal bipyramid state', is shown in \fig{bloch_5}(a).  This state
has the form $| \psi_{5} \rangle = 1 / \sqrt{2} ( | S_{1} \rangle + |
S_{4} \rangle )$, and a simple calculation yields that it has three
CPPs that coincide with the equatorial MPs, giving an entanglement of
$E_{\text{G}} ( | \psi_{5} \rangle ) = \log_2 ( 16 / 5 ) \approx
1.68$.

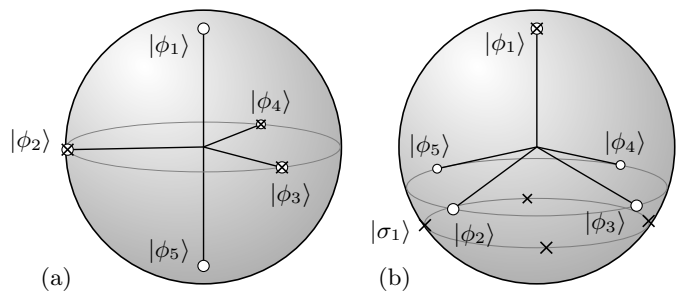
\begin{figure}[hb]
  \begin{center}
    \begin{overpic}[scale=.5]{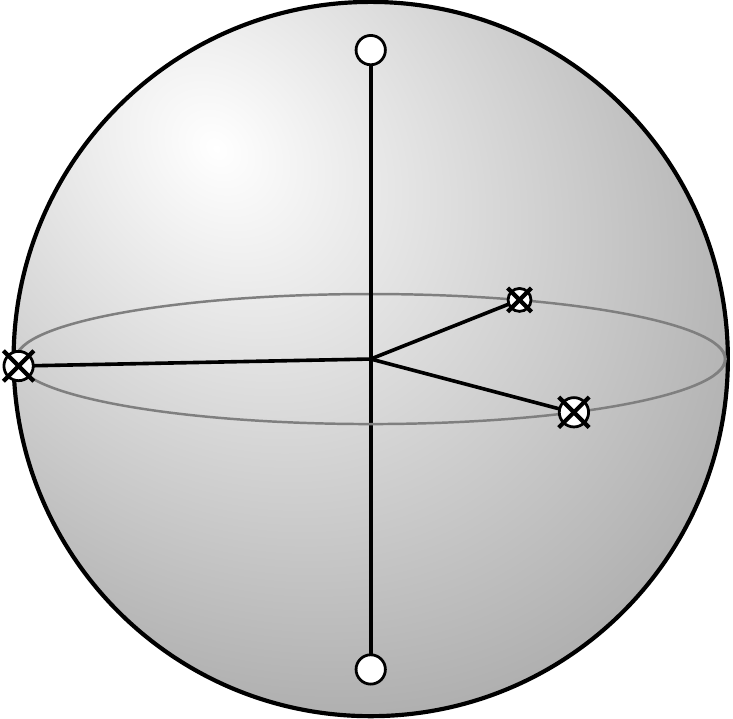}
      \put(-7,0){(a)}
      \put(31,83){$| \phi_1 \rangle$}
      \put(-18,49){$| \phi_2 \rangle$}
      \put(74,30){$| \phi_3 \rangle$}
      \put(67,63){$| \phi_4 \rangle$}
      \put(31,10){$| \phi_5 \rangle$}
    \end{overpic}
    \hspace{5mm}
    \begin{overpic}[scale=.5]{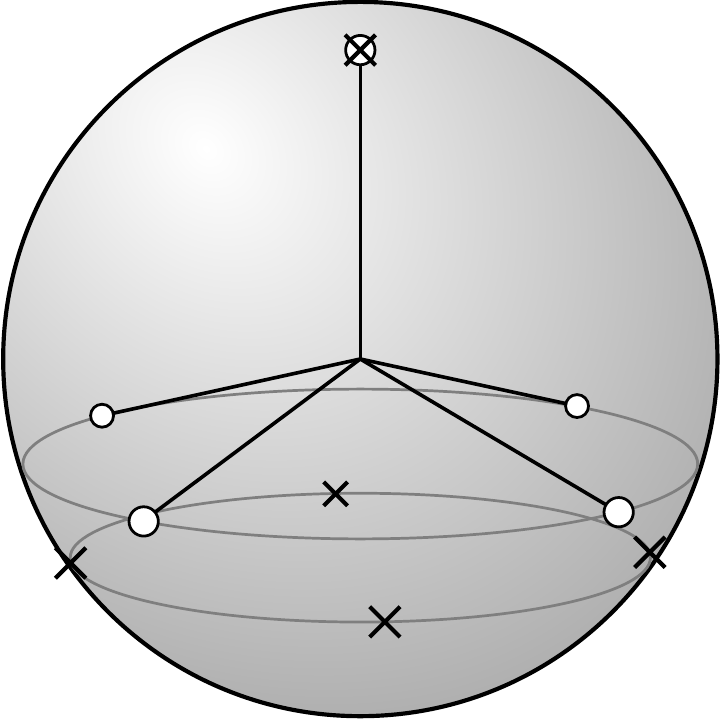}
      \put(-7,0){(b)}
      \put(31,85){$| \phi_1 \rangle$}
      \put(20,16){$| \phi_2 \rangle$}
      \put(64,20){ $| \phi_3 \rangle$}
      \put(74,49){$| \phi_4 \rangle$}
      \put(4,48){$| \phi_5 \rangle$}
      \put(-10,17){$| \sigma_1 \rangle$}
    \end{overpic}
  \end{center}
  \caption{\label{bloch_5} The distribution (a) shows the `trigonal
    bipyramid state', but the conjectured solution of the Majorana
    problem is the `square pyramid state', shown in (b).}
\end{figure}

However, a numerical search among all positive symmetric states yields
states with higher entanglement.  Our numerics indicate that the
maximally entangled state is the `square pyramid state' shown in
\fig{bloch_5}(b).  This state has five CPPs, one on the north pole and
the other four lying in a horizontal plane slightly below the plane
with the MPs.  The form of this state is
\begin{equation}\label{5_opt_form}
  | \Psi_{5} \rangle = \tfra{1}{\sqrt{1 + A^2}} | S_{0} \rangle +
  \tfra{A}{\sqrt{1 + A^2}} | S_{4} \rangle \enspace .
\end{equation}
Its MPs are
\begin{equation}\label{5_opt_maj}
  \begin{split}
    | \phi_{1} \rangle & = | 0 \rangle \enspace , \\
    | \phi_{2,3,4,5} \rangle & = \alpha | 0 \rangle + \E^{\I \kappa}
    \sqrt{1 - \alpha^2} | 1 \rangle \enspace ,
  \end{split}
\end{equation}
with $\kappa = \tfra{\pi}{4}, \tfra{3 \pi}{4}, \tfra{5 \pi}{4},
\tfra{7 \pi}{4}$, and the CPPs are
\begin{equation}\label{5_opt_cpp}
  \begin{split}
    | \sigma_{1} \rangle & = | 0 \rangle \enspace , \\
    | \sigma_{2,3,4,5} \rangle & = x | 0 \rangle + k \sqrt{1 - x^2}
    | 1 \rangle \enspace ,
  \end{split}
\end{equation}
with $k = 1, \I, -1, -\I$. The exact values can be determined
analytically by solving quartic equations.  The $x \in (0,1)$ of
\eq{5_opt_cpp} is given by the real root of $4 x^4 + 4 x^3 + 4 x^2 - x
- 1 = 0$, and this can be used to calculate $A = ( 1 - x^5 ) / (
\sqrt{ 5 } x (1 - x^2)^2 )$.  With the substitution $a = \alpha^2
\in (0,1)$ the value of $\alpha$ is given by the real root of $(5 A^2
- 1) a^4 + 4 a^3 - 6 a^2 + 4 a - 1 = 0$.  Approximate values of these
quantities are:
\begin{equation}\label{5_opt_form_approx}
  x \approx 0.46657 \enspace , \quad  A \approx 1.53154 \enspace ,
  \quad \alpha \approx 0.59229 \enspace .
\end{equation}
The entanglement is $E_{\text{G}} ( | \Psi_{5} \rangle ) = \log_2 ( 1
+ A^2 ) \approx 1.74$, which is considerably higher than that of
$E_{\text{G}} ( | \psi_{5} \rangle )$. We remark that the `center of
mass' of the five MPs of $| \Psi_{5} \rangle$ does not coincide with
the sphere's origin, thus ruling out the corresponding spin-$5/2$
state as an anticoherent spin state, as defined in \cite{Zimba06}.

\subsection{Six qubits}\label{majorana_six}

The regular octahedron, a Platonic solid, is the unique solution to
T\'{o}th's and Thomson's problem for six points.  The corresponding
`octahedron state' was numerically confirmed to solve the Majorana
problem for six qubits.

\begin{figure}[hb]
  \begin{center}
    \begin{overpic}[scale=.5]{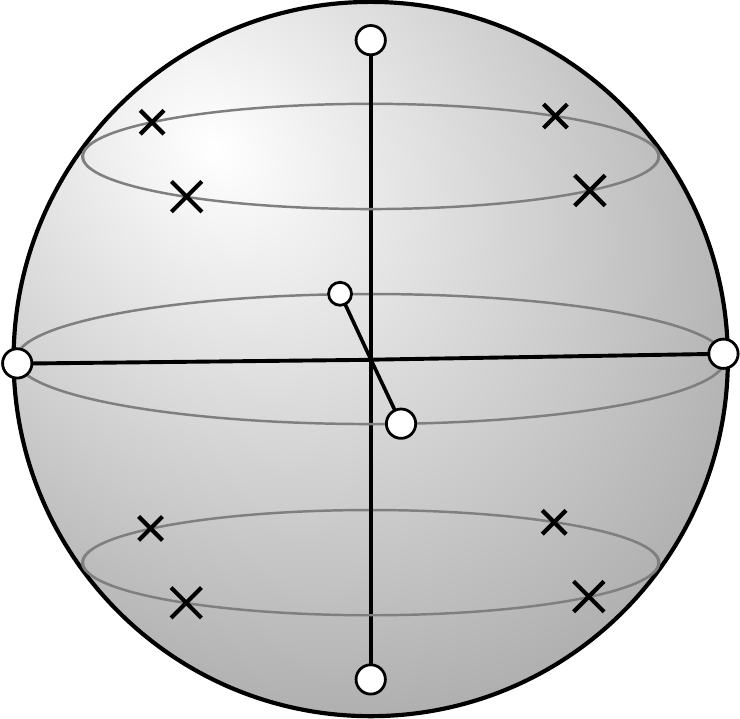}
      \put(-7.5,0){(a)}
      \put(32,84){$| \phi_1 \rangle$}
      \put(32,9){$| \phi_2 \rangle$}
      \put(-17,48){$| \phi_3 \rangle$}
      \put(57,31){$| \phi_4 \rangle$}
      \put(79,56){$| \phi_5 \rangle$}
      \put(27,56){$| \phi_6 \rangle$}
    \end{overpic}
    \hspace{5mm}
    \begin{overpic}[scale=.5]{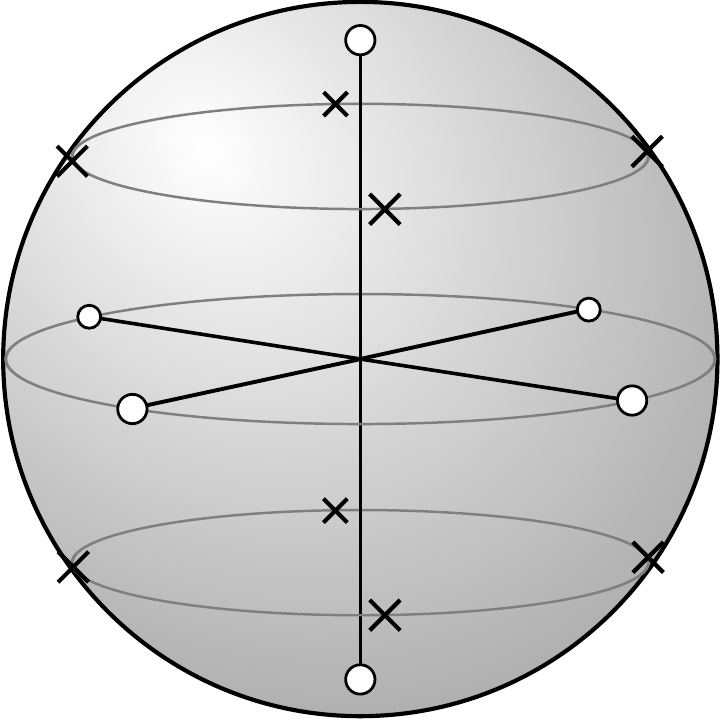}
      \put(-7.5,0){(b)}
      \put(32,89){$| \phi_1 \rangle$}
      \put(31,8){$| \phi_2 \rangle$}
      \put(11,32){$| \phi_3 \rangle$}
      \put(78,33){$| \phi_4 \rangle$}
      \put(77,62){$| \phi_5 \rangle$}
      \put(7,62){$| \phi_6 \rangle$}
    \end{overpic}
  \end{center}
  \caption{\label{bloch_6_1} Two possible orientations of the
    `octahedron state'.}
\end{figure}
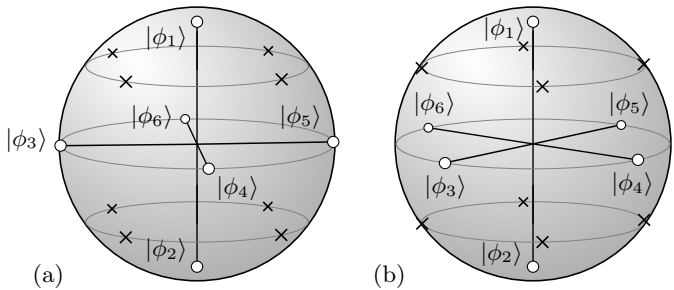

The straightforward orientation shown in \fig{bloch_6_1}(a) has the
form $| \Psi_{6} ' \rangle = 1/{\sqrt{2}} ( | S_{1} \rangle - | S_{5}
\rangle )$, and its MPs are
\begin{equation}\label{6_alpha1_maj}
  \begin{split}
    & | \phi_{1} \rangle = | 0 \rangle \, , \quad
    | \phi_{2} \rangle = | 1 \rangle \enspace , \\
    & | \phi_{3,4,5,6} \rangle = \tfra{1}{\sqrt{2}} \big( | 0 \rangle
    + k | 1 \rangle \big) \enspace ,
  \end{split}
\end{equation}
with $k = 1, \I, -1, -\I$.  $| \Psi_{6} ' \rangle$ can be turned into
the positive state $| \Psi_{6} \rangle = 1 / \sqrt{2} ( | S_{1}
\rangle + | S_{5} \rangle )$ by means of an $R^{\, \text{s}}_z (\pi /
4)$ rotation.  The CPPs can be obtained from this state in the same
way as for the tetrahedron state. Being a Platonic solid, the MP
distribution of \fig{bloch_6_1}(b) is left invariant under a finite
subgroup of rotation operations on the sphere.  From Lemma
\ref{cpp_mer}, the intersection of the permissible locations of the
CPPs is found to be the eight points lying at the center of each face
of the octahedron, forming a cube inside the Majorana sphere.
\begin{equation}\label{6_sigma_2}
  \begin{split}
    | \sigma_{1,2,3,4} \rangle & = \sqrt{ \tfra{\sqrt{3}+1}{2
        \sqrt{3}}} \, | 0 \rangle + k \sqrt{ \tfra{\sqrt{3}-1}{2
        \sqrt{3}}} \, | 1 \rangle \enspace , \\
    | \sigma_{5,6,7,8} \rangle & = \sqrt{ \tfra{\sqrt{3}-1}{2
        \sqrt{3}}} \, | 0 \rangle + k \sqrt{ \tfra{\sqrt{3}+1}{2
        \sqrt{3}}} \, | 1 \rangle \enspace ,
  \end{split}
\end{equation}
with $k = 1,\I,-1,-\I$.  In contrast to the tetrahedron state, where
the MPs and CPPs overlap, the CPPs of the octahedron state lie as far
away from the MPs as possible.  This is plausible, because in the case
of the octahedron \eq{maj_problem} is zero at the location of any MP,
due to the MPs forming diametrically opposite pairs.  The amount of
entanglement is $E_{\text{G}} ( | \Psi_{6} \rangle ) = \log_2 (9/2)
\approx 2.17$.

\subsection{Seven qubits}\label{majorana_seven}

For seven points, the solutions to the two classical problems become
fundamentally different for the first time.  T\'{o}th's problem is
solved by two triangles asymmetrically positioned about the equator
and the remaining point at the north pole \cite{Erber91}, or (1-3-3)
in the F{\"o}ppl notation \cite{Whyte52}.  Thomson's problem is solved
by the vertices of a pentagonal dipyramid
\cite{Ashby86,Marx70,Erber91}, where five points lie on an equatorial
pentagon and the other two on the poles. The latter is also
numerically found to be the solution to the Majorana problem.

\begin{figure}[ht]
  \begin{center}
    \begin{overpic}[scale=.5]{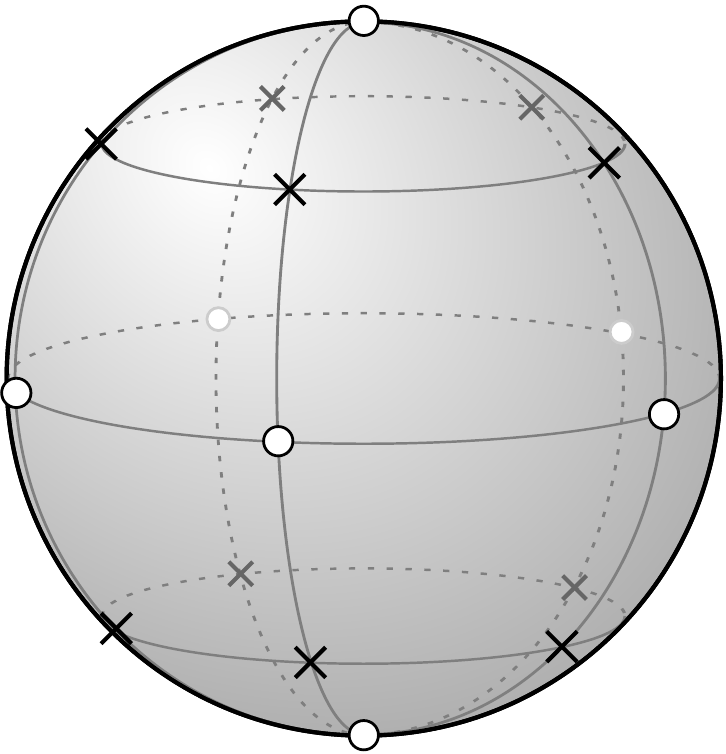}
      \put(38,103){$| \phi_1 \rangle$}
      \put(34,-8){$| \phi_2 \rangle$}
      \put(-16,50){$| \phi_3 \rangle$}
      \put(20,33){$| \phi_4 \rangle$}
      \put(77,35){$| \phi_5 \rangle$}
      \put(70,62){$| \phi_6 \rangle$}
      \put(13,62){$| \phi_7 \rangle$}
      \put(2,88){$| \sigma_1 \rangle$}
      \put(3,6){$| \sigma_2 \rangle$}
    \end{overpic}
  \end{center}
  \caption{\label{bloch_7} MPs and CPPs of the `pentagonal dipyramid
    state'. The ten CPPs are equidistantly located on two circles
    above and below the equator.}
\end{figure}
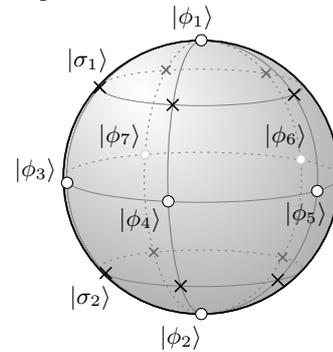

The `pentagonal dipyramid state', shown in \fig{bloch_7}, has the form
$| \Psi_{7} \rangle = 1/{\sqrt{2}} ( | S_{1} \rangle + | S_{6} \rangle
)$, and its MPs are
\begin{equation}\label{7_maj}
  \begin{split}
    & | \phi_{1} \rangle = | 0 \rangle \, , \quad
    | \phi_{2} \rangle = | 1 \rangle \enspace , \\
    & | \phi_{3,4,5,6,7} \rangle = \tfra{1}{\sqrt{2}} \big( | 0
    \rangle + \E^{\I \kappa} | 1 \rangle \big) \enspace ,
  \end{split}
\end{equation}
with $\kappa = 0, \tfra{2 \pi}{5}, \tfra{4 \pi}{5}, \tfra{6 \pi}{5},
\tfra{8 \pi}{5}$.  The CPPs of this positive state can be determined
analytically by choosing a suitable parametrization.  With $x :=
\cos^2 \theta$ the position of $| \sigma_1 \rangle = \co_{\theta} | 0
\rangle + \si_{\theta} | 1 \rangle$ and $| \sigma_2 \rangle =
\si_{\theta} | 0 \rangle + \co_{\theta} | 1 \rangle$ is determined by
the real root of the cubic equation $49 x^3 + 165 x^2 - 205 x + 55 =
0$ in the interval $[0,\tfra{1}{2}]$.  The approximate amount of
entanglement is $E_{\text{G}} ( | \Psi_{7} \rangle ) \approx 2.299$.

\subsection{Eight qubits}\label{majorana_eight}

For eight points, T\'{o}th's problem is solved by the cubic antiprism,
a cube with one face rotated around 45$^\circ$ and where the distances
between neighboring vertices are all the same. The solution to
Thomson's problem is obtained by stretching this cubic antiprism along
the Z-axis, thereby introducing two different nearest-neighbor
distances and further lowering symmetry \cite{Marx70,Ashby86,Erber91}.
One would expect that a similar configuration solves the Majorana
problem too, but, surprisingly, this is not the case. The `asymmetric
pentagonal dipyramid' shown in \fig{bloch_8}(b) is numerically found
to have the highest amount of entanglement. An analytic form of this
positive state is not known, but it can be numerically approximated to
a very high precision.  The state is $| \Psi_{8} \rangle \approx 0.672
| S_{1} \rangle + 0.741 | S_{6} \rangle$, and its entanglement is
$E_{\text{G}} ( | \Psi_{8} \rangle ) \approx 2.45$.  For comparison,
the regular cube yields $E_{\text{G}} ( | \psi_{\text{cube}} \rangle )
= \log_2 (24/5) \approx 2.26$.  Interestingly, the positive state $|
\Psi_{8} \rangle$ has a higher amount of entanglement than any state
of the antiprism form, all of which are non-positive.  Furthermore,
two of the MPs of $| \Psi_{8} \rangle$ coincide, akin to the W state
of three qubits.

\begin{figure}[ht]
  \begin{center}
    \begin{overpic}[scale=.5]{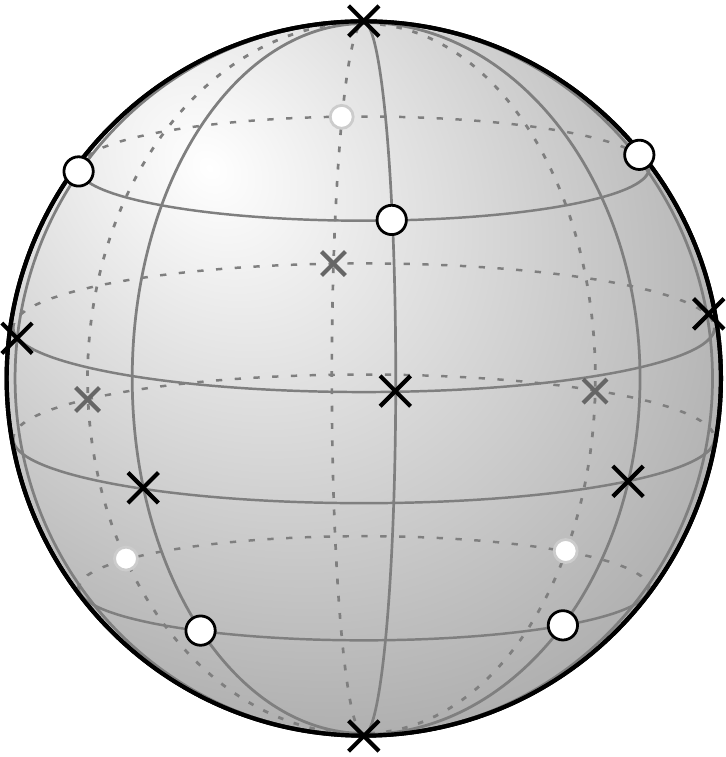}
      \put(-7.5,0){(a)}
    \end{overpic}
    \hspace{5mm}
    \begin{overpic}[scale=.5]{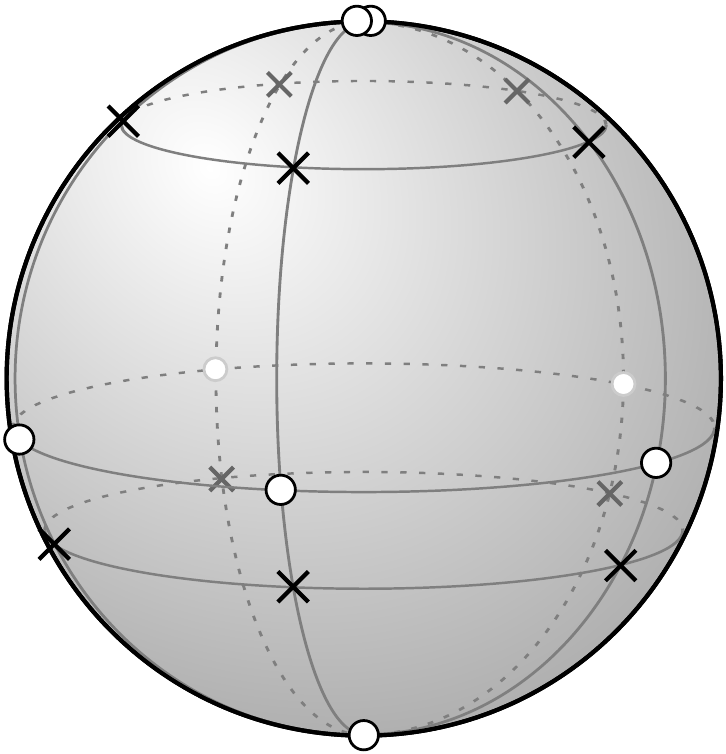}
      \put(-7.5,0){(b)}
    \end{overpic}
  \end{center}
  \caption{\label{bloch_8} The maximally entangled antiprism state is
    shown in (a), while the `asymmetric pentagonal dipyramid state' in
    (b) is conjectured to be the maximally entangled state.}
\end{figure}

$| \Psi_{8} \rangle$ is positive, with two MPs lying on the north
pole, one on the south pole and the other five on a circle below the
equator.  The exact inclination of this circle as well as the
inclination of the two circles with the CPPs is not known, but can be
approximated numerically.

\subsection{Nine qubits}\label{majorana_nine}

For nine points, the solutions to T\'{o}th's and Thomson's problems
are slightly different manifestations of the same geometric form
(3-3-3) of neighboring triangles being asymmetrically positioned.
This is also known as a triaugmented triangular prism.  In contrast to
this, the Majorana problem is numerically solved by $| \Psi_{9}
\rangle = 1 / \sqrt{2} ( | S_{2} \rangle + | S_{7} \rangle )$, shown
in \fig{bloch_9}.  This is again a positive state with a rotational
Z-axis symmetry and with coinciding MPs.

\begin{figure}
  \begin{center}
    \begin{overpic}[scale=.5]{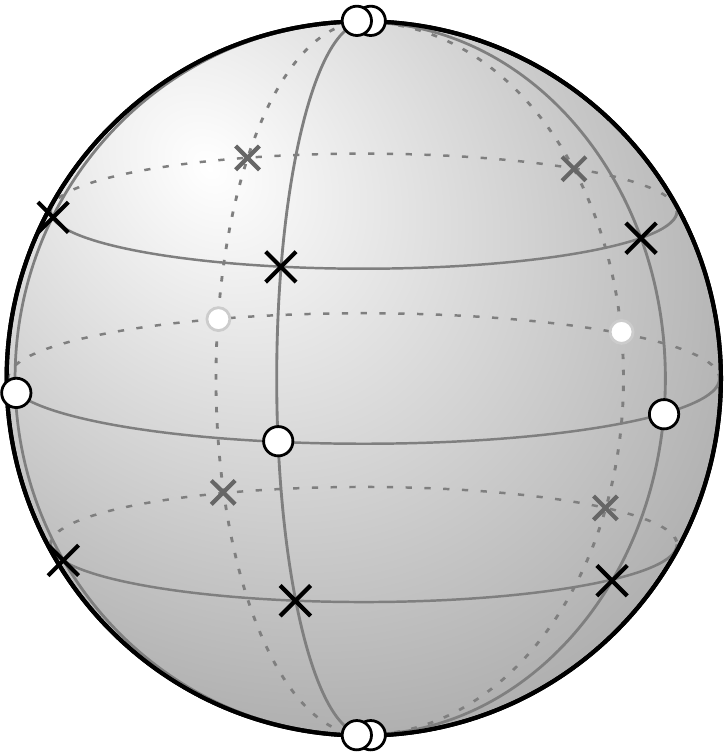}
    \end{overpic}
  \end{center}
  \caption{\label{bloch_9} The `pentagonal dipyramid state' with both
    of the poles occupied by two MPs.}
\end{figure}

The CPPs can be found analytically in the same way as for seven
qubits.  With the substitution $x := \cos^2 \theta$, the inclination
$\theta$ of the CPPs in the northern hemisphere (or $\pi - \theta$ for
the CPPs in the southern hemisphere) follows from the real root of $81
x^3 + 385 x^2 - 245 x + 35 = 0$ in the interval $[0,0.3]$.  The
approximate amount of entanglement is $E_{\text{G}} ( | \Psi_{9}
\rangle ) \approx 2.554$.

\subsection{Ten qubits}\label{majorana_ten}

The solution to T\'{o}th's problem is an arrangement of the form
(2-2-4-2), while Thomson's problem is solved by the gyroelongated
square bipyramid, a deltahedron that arises from a cubic antiprism by
placing square pyramids on each of the two square surfaces.

The ten-qubit case is distinct in two respects. It is the first case
where the numerically determined positive solution is not rotationally
symmetric around any axis.  Furthermore, we found non-positive states
with higher entanglement than the conjectured solution for positive
states.

A numerical search returns a state of the form $| \Psi_{10} \rangle =
\alpha | S_{0} \rangle + \beta | S_{4} \rangle + \gamma | S_{9}
\rangle$ as the positive state with the highest entanglement, namely
$E_{\text{G}} ( | \Psi_{10} \rangle ) \approx 2.6798$.  From Lemma
\ref{rot_symm} it is clear that this state is not rotationally
symmetric around the Z-axis.  The MP distribution is shown in
\fig{bloch_10}(a). The state has only three CPPs, which are all
positive (cf. Theorem \ref{maj_max_pos_zero}), but there are six other
local maxima of $g( \sigma )$ with values close to the CPPs.  Their
positions are shown by dashed crosses in \fig{bloch_10}(a).  While the
total MP distribution is not rotationally symmetric around the Z-axis,
one would expect from the numerical results that the MPs form two
horizontal planes, one with five MPs and another with four MPs, with
equidistantly spread out MPs. However, this is not the case, as the
locations of the MPs deviate by small, but significant, amounts from
this simple form.

\begin{figure}[hb]
    \begin{center}
      \begin{overpic}[scale=.35]{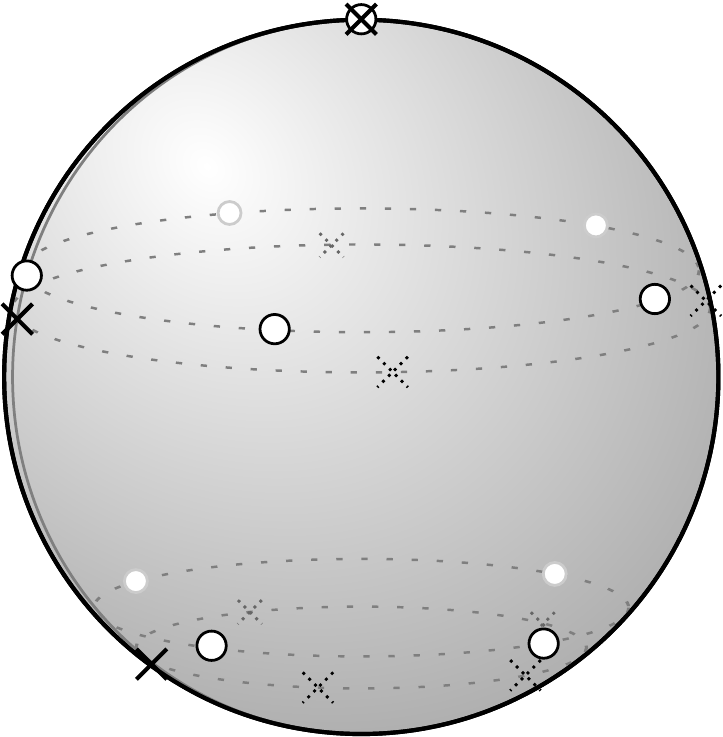}
        \put(-7.5,0){(a)}
      \end{overpic}
      \hspace{2mm}
      \begin{overpic}[scale=.35]{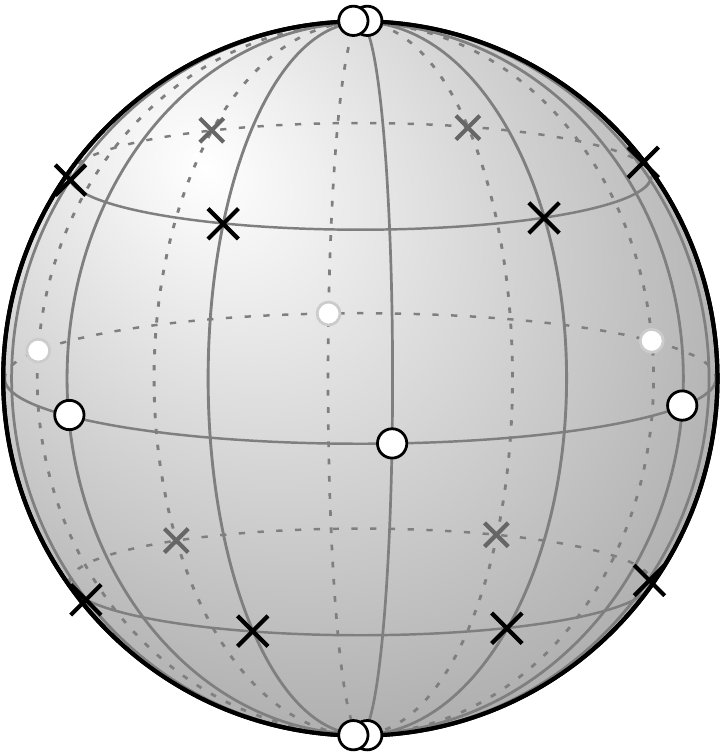}
        \put(-7.5,0){(b)}
      \end{overpic}
      \hspace{2mm}
      \begin{overpic}[scale=.35]{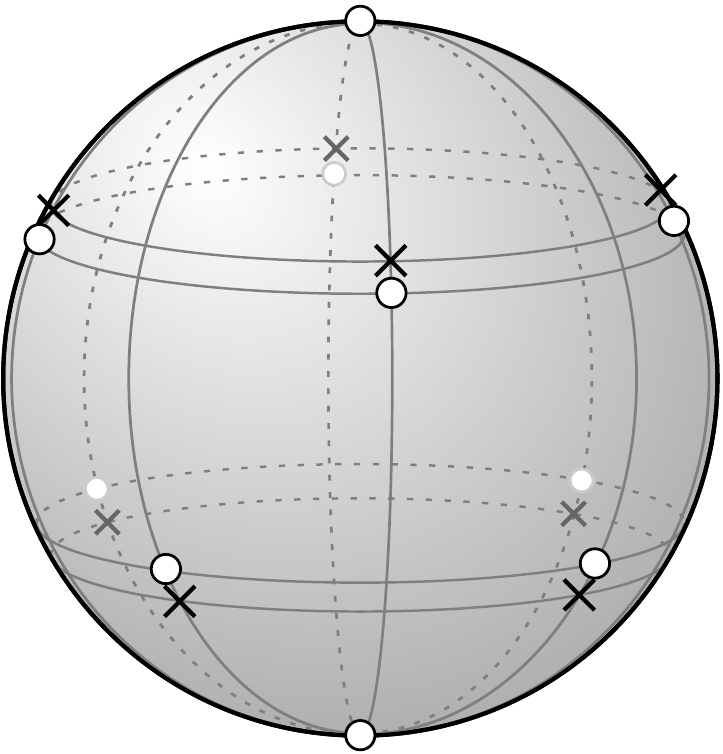}
        \put(-7.5,0){(c)}
      \end{overpic}
    \end{center}
    \caption{\label{bloch_10} The numerically determined maximally
      entangled positive state is shown in (a). A similarly highly
      entangled positive state with a rotational symmetry is shown in
      (b). The candidate for the general case is shown in (c).}
\end{figure}

Interestingly, there is a fully rotationally symmetric positive state
that comes very close to $| \Psi_{10} \rangle$ in terms of
entanglement. Its straightforward form is $| \psi_{10} \rangle = 1 /
\sqrt{2} ( | S_{2} \rangle + | S_{8} \rangle )$, as displayed in
\fig{bloch_10}(b).  The 12 CPPs of this state are easily found as the
solutions of a quadratic equation. The two positive CPPs are
\begin{equation}\label{10_sigma}
  \begin{split}
    | \sigma_{1} \rangle & = \tfra{1}{\sqrt{3 - \sqrt{3}}} \, | 0
    \rangle + \tfra{1}{\sqrt{3 + \sqrt{3}}} \, | 1 \rangle \enspace , \\
    | \sigma_{2} \rangle & = \tfra{1}{\sqrt{3 + \sqrt{3}}} \, | 0
    \rangle + \tfra{1}{\sqrt{3 - \sqrt{3}}} \, | 1 \rangle \enspace ,
  \end{split}
\end{equation}
and the entanglement is $E_{\text{G}} ( | \psi_{10} \rangle ) = \log_2
(32/5) \approx 2.6781$.  This is less than $0.1 \%$ difference from $|
\Psi_{10} \rangle$.

The solution to Thomson's problem, recast as a quantum state of the
form $| \Psi_{10} ' \rangle = \alpha | S_{1} \rangle + \beta | S_{5}
\rangle - \alpha | S_{9} \rangle$, is not positive and has an
entanglement of $E_{\text{G}} \approx 2.7316$.  From numerics one can
see that the entanglement of this state can be further increased by
slightly modifying the coefficients, arriving at a state with eight
CPPs and an entanglement of $E_{\text{G}} ( | \Psi_{10} ' \rangle )
\approx 2.7374$.  The state is shown in \fig{bloch_10}(c), and we
propose it as a candidate for the maximally entangled symmetric state
of ten qubits.

\subsection{Eleven qubits}\label{majorana_eleven}

The solution to T\'{o}th's problem is a pentagonal antiprism with a
pentagonal pyramid on one of the two pentagonal surfaces, or
(1-5-5). The solution to Thomson's problem is of the form (1-2-4-2-2).
Analogous to the ten-qubit case, the numerically found positive state
of 11 qubits with maximal entanglement is not rotationally
symmetric. The state, shown in \fig{bloch_11}(a), is of the form $|
\Psi_{11} \rangle = \alpha | S_{1} \rangle + \beta | S_{5} \rangle +
\gamma | S_{10} \rangle$, and its entanglement is $E_{\text{G}} ( |
\Psi_{11} \rangle ) \approx 2.77$.  The state has only two CPPs, but
there exist seven more local maxima with values close to the CPPs.

\begin{figure}[hb]
  \begin{center}
    \begin{overpic}[scale=.5]{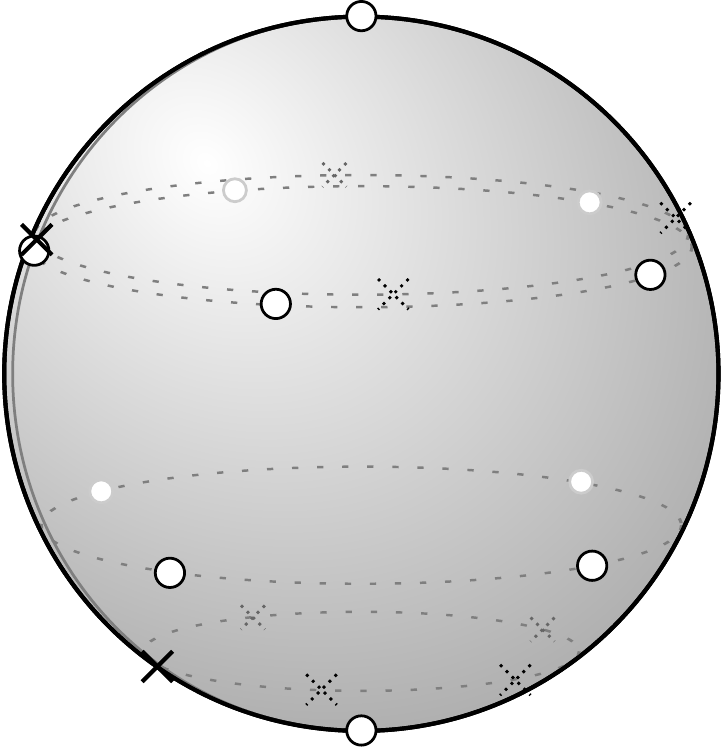}
      \put(-7.5,0){(a)}
    \end{overpic}
    \hspace{5mm}
    \begin{overpic}[scale=.5]{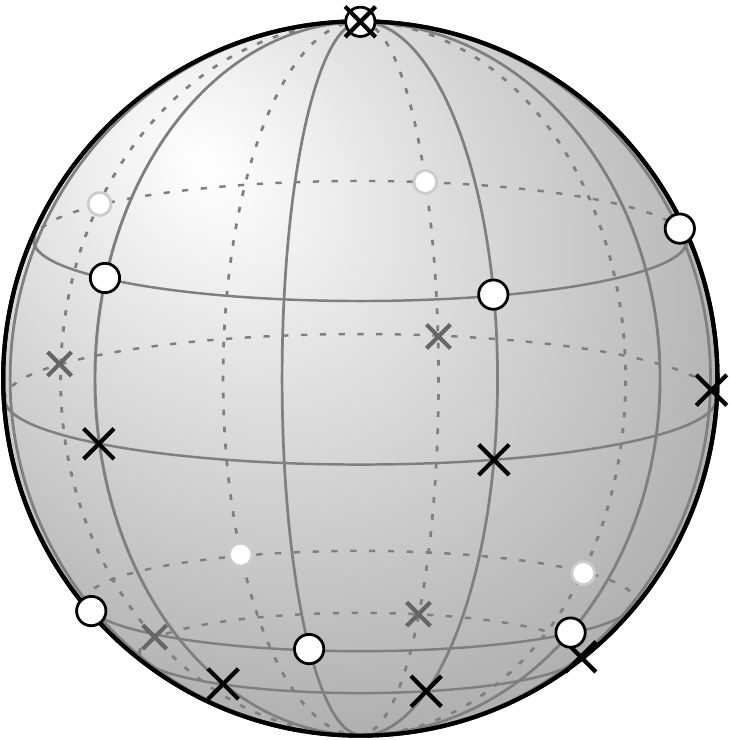}
      \put(-7.5,0){(b)}
    \end{overpic}
  \end{center}
  \caption{\label{bloch_11} The conjectured maximally entangled
    positive state of 11 qubits is shown in (a), while the candidate
    for the general case is shown in (b).}
\end{figure}

The solution to T\'{o}th's problem, which is of the form $| \psi_{11}
' \rangle = \alpha | S_{0} \rangle + \beta | S_{5} \rangle - \gamma |
S_{10} \rangle$, yields very low entanglement, but by modifying the
coefficients of this non-positive state one can find a state $|
\Psi_{11} ' \rangle$ which is even more entangled than $| \Psi_{11}
\rangle$.  As shown in \fig{bloch_11}(b), the state is rotationally
symmetric around the Z-axis and has 11 CPPs.  The entanglement is
$E_{\text{G}} ( | \Psi_{11} ' \rangle ) \approx 2.83$, making the
state the potentially maximally entangled state of 11 qubits.

\subsection{Twelve qubits}\label{majorana_twelve}

For 12 points, both of the classical problems are solved by the
icosahedron, a Platonic solid.  Because the icosahedron cannot be cast
as a positive state, the numerical search for positive states yields a
different state of the form $| \Psi_{12} ' \rangle = \alpha | S_{1}
\rangle + \beta | S_{6} \rangle + \alpha | S_{11} \rangle$. From
\fig{bloch_12}(a) it can be seen that this state can be thought of as
an icosahedron with one circle of MPs rotated by 36$^\circ$ so that it
is aligned with the MPs of the other circle.  There are three circles
of CPPs with five in each circle. One of these planes coincides with
the equator, so that $\vert \sigma \rangle = 1 / \sqrt{2}
(\vert0\rangle + \vert1\rangle)$ is a trivial CPP. Nevertheless, the
exact location of some of the MPs and CPPs are unknown. The
approximate entanglement is $E_{\text{G}} ( | \Psi_{12} ' \rangle )
\approx 2.99$.

\begin{figure}[ht]
  \begin{center}
    \begin{overpic}[scale=.5]{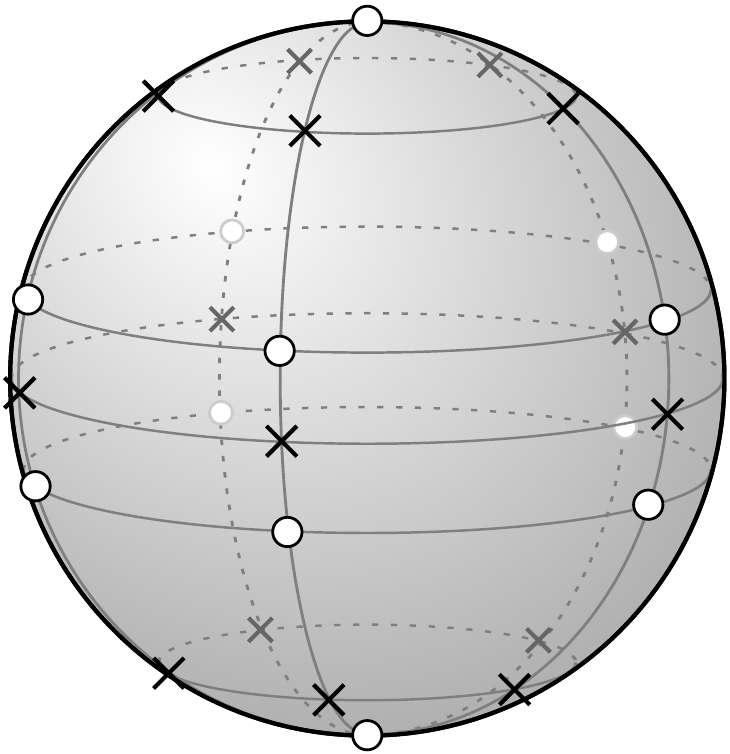}
      \put(-7.5,0){(a)}
    \end{overpic}
    \hspace{5mm}
    \begin{overpic}[scale=.5]{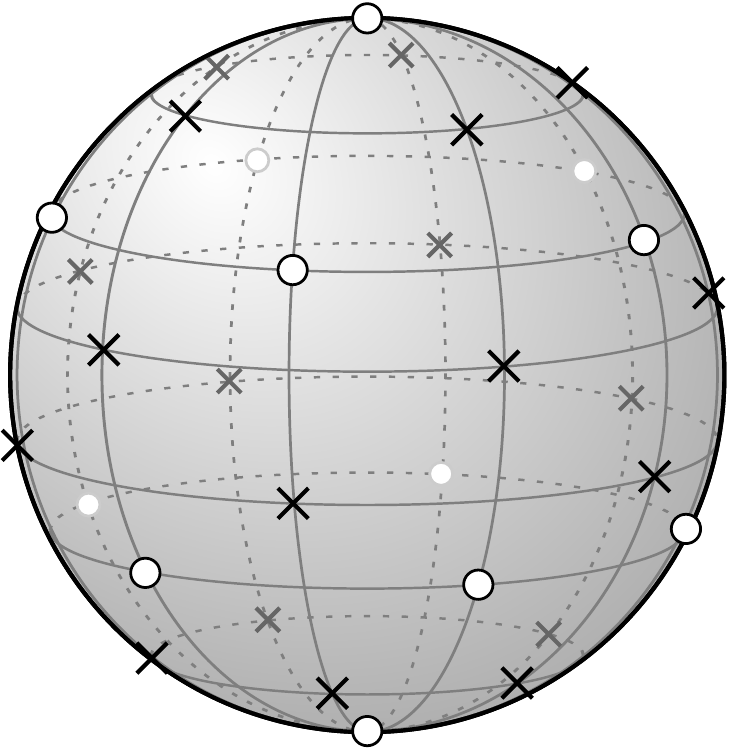}
      \put(-7.5,0){(b)}
    \end{overpic}
  \end{center}
  \caption{\label{bloch_12} An orientation of the maximally entangled
    positive symmetric state of 12 qubits is shown in (a).  The
    icosahedron state, shown in (b), is considered to be the maximally
    entangled of all symmetric 12 qubit states.}
\end{figure}

Due to the high symmetry present in Platonic solids, the `icosahedron
state' is a strong candidate for maximal symmetric entanglement of
twelve qubit states.  The state can be cast with real coefficients $|
\Psi_{12} \rangle = \tfra{\sqrt{7}}{5} | S_{1} \rangle -
\tfra{\sqrt{11}}{5} | S_{6} \rangle - \tfra{\sqrt{7}}{5} | S_{11}
\rangle$, and its MPs can be easily derived from the known angles and
distances in the icosahedron.
\begin{equation}\label{12_maj}
  \begin{split}
    & | \phi_{1} \rangle = | 0 \rangle \, , \quad
    | \phi_{12} \rangle = | 1 \rangle \enspace , \\
    & | \phi_{2,3,4,5,6} \rangle = \sqrt{ \tfra{3+\sqrt{5}}{5
        +\sqrt{5}}} \, | 0 \rangle + \E^{\I \kappa}
    \sqrt{\tfra{2}{5+\sqrt{5}}} \, | 1 \rangle \enspace , \\
    & | \phi_{7,8,9,10,11} \rangle = \sqrt{ \tfra{2}{5+\sqrt{5}}}
    \, | 0 \rangle + \E^{\I (\kappa + \pi / 5 )} \sqrt{
      \tfra{3+\sqrt{5}}{5 +\sqrt{5}}} \, | 1 \rangle \enspace ,
  \end{split}
\end{equation}
with $\kappa = 0, \tfra{2 \pi}{5}, \tfra{4 \pi}{5}, \tfra{6 \pi}{5},
\tfra{8 \pi}{5}$. From numerics and from the Z-axis rotational
symmetry, it is evident that there are 20 CPPs, one at the center of
each face of the icosahedron. Equivalent to the six-qubit case, the
MPs appear as diametrically opposite pairs, forcing the CPPs to be as
remote from the MPs as possible.  The CPPs are
\begin{equation}\label{12_sigma}
  \begin{split}
    | \sigma_{1,\ldots ,5}  \rangle & = \aaa_{+} | 0 \rangle +
    \E^{\I (\kappa + \pi / 5 )} \, \aaa_{-} | 1
    \rangle \enspace , \\
    | \sigma_{6,\ldots ,10} \rangle & = \bbb_{+} | 0 \rangle +
    \E^{\I (\kappa + \pi / 5 )} \, \bbb_{-} | 1
    \rangle \enspace , \\
    | \sigma_{11,\ldots,15} \rangle & = \bbb_{-} | 0 \rangle +
    \E^{\I \kappa} \, \bbb_{+} | 1 \rangle \enspace , \\
    | \sigma_{16,\ldots,20} \rangle & = \aaa_{-} | 0 \rangle +
    \E^{\I \kappa} \, \aaa_{+} | 1 \rangle \enspace ,
  \end{split}
\end{equation}
with $\kappa = 0, \tfra{2 \pi}{5}, \tfra{4 \pi}{5}, \tfra{6
  \pi}{5}, \tfra{8 \pi}{5}$ and
\begin{equation}\label{12_sigma_coeff}
  \begin{split}
    \aaa_{\pm} & = \sqrt{\frac{1}{2} \pm \frac{1}{2}
      \sqrt{\frac{5+2 \sqrt{5}}{15}}} \enspace , \\
    \bbb_{\pm} & = \sqrt{\frac{1}{2} \pm \frac{1}{2}
      \sqrt{\frac{5-2 \sqrt{5}}{15}}} \enspace .
  \end{split}
\end{equation}
With the knowledge of the exact positions of the MPs and CPPs, the
entanglement of the icosahedron state can be calculated as
$E_{\text{G}} ( | \Psi_{12} \rangle ) = \log_2 (243/28) \approx
3.1175$.  \Fig{icosahedronpic} shows a spherical plot of the overlap
function $g( \sigma ) = | \langle \Psi_{12} | \sigma \rangle^{\otimes
  12} |$ from the same viewpoint as in \fig{bloch_12}(b).  Due to
their diametrically opposite pairs, the MPs coincide with the zeros in
this plot. The CPPs can be identified as the maxima of $g( \sigma )$.

\begin{figure}[ht]
  \begin{center}
    \begin{overpic}[scale=.20]{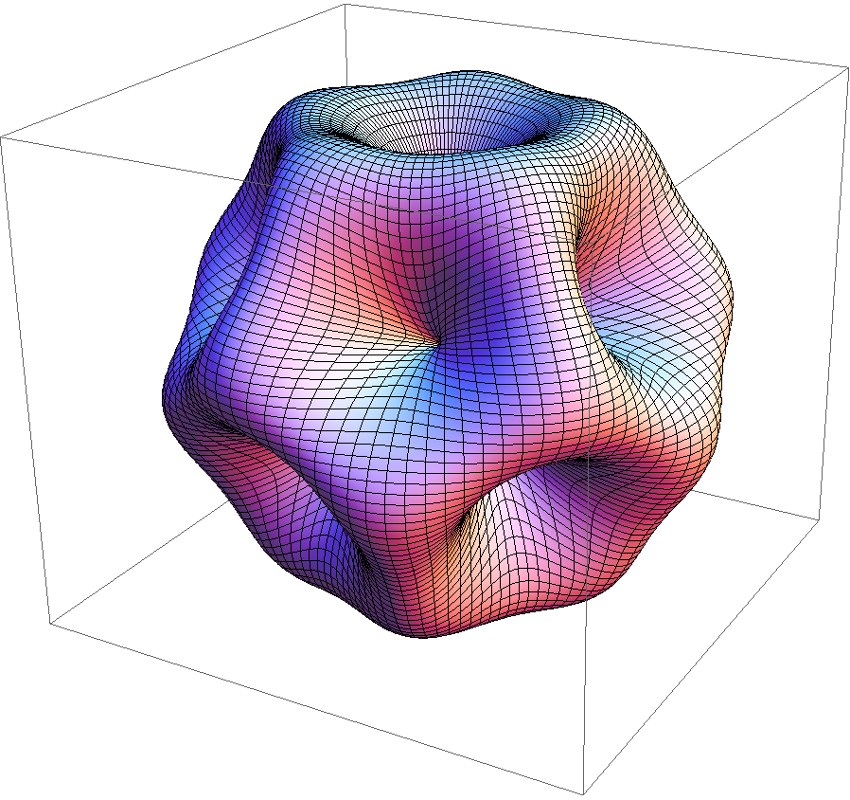}
    \end{overpic}
  \end{center}
  \caption{\label{icosahedronpic} (color online) A spherical plot of
    the overlap function $g( \sigma ) = |\langle \Psi_{12} | \sigma
    \rangle^{\otimes 12}|$ for the icosahedron state $| \Psi_{12}
    \rangle$.}
\end{figure}

\section{Discussion}\label{discussion}

The MP distribution of highly entangled states can be explained with
the overlap function $g(\sigma) = |\langle \psi | \sigma
\rangle^{\otimes n}|$ seen in \fig{icosahedronpic}.
\app{normalization_bloch} states that the integration volume of
$g(\sigma)^2$ over the sphere is the same for all symmetric
states. Therefore a bunching of the MPs in a small area of the sphere
would lead to high values of $g(\sigma)^2$ in that area, thus lowering
the entanglement.  This explains the tendency of MPs to spread out as
far as possible, as it is seen for the classical problems.  However,
there also exist highly entangled states where two or more MPs
coincide (as seen for $n =3,8,9$).  This is intriguing because such
configurations are the least optimal ones for classical
problems. Again, this can be explained with the constant integration
volume of $g(\sigma)^2$. Because of $g(\sigma)^2 \propto \prod_i
|\langle \phi_i | \sigma \rangle |^2$, the zeros of $g(\sigma)^2$ are
the diametrically opposite points (antipodes) of the MPs and therefore
a lower number of \emph{different} MPs leads to a lower number of
zeros in $g(\sigma)^2$. This can lead to the integration volume being
more evenly distributed over the sphere, thus yielding a higher amount
of entanglement.

Excluding the Dicke states with their infinite amount of CPPs, one
observes that highly entangled states tend to have a large number of
CPPs.  The prime example for this is the case of five qubits, where
the classical solution with only three CPPs is less entangled than the
`square pyramid' state that has five CPPs.  In Theorem
\ref{maj_max_pos_zero} it was shown that $2n-4$ is an upper bound on
the number of CPPs of positive symmetric $n$ qubit states.  For all of
our numerically determined maximally entangled states -- including the
non-positive ones -- this bound is obeyed, and for most states the
number of CPPs is close to the bound ($n=5,8$) or coincides with it
($n = 4,6,7,12$).  This raises the question whether this bound also
holds for general symmetric states. To date, neither proof nor
counterexample is known.

When viewing the $n$ MPs of a symmetric state as the edges of a
polyhedron, Euler's formula for convex polyhedra yields the upper
bound $2n-4$ on the number of faces. This bound is strict if no pair
of MPs coincides and all polyhedral faces are triangles.
Intriguingly, this bound is the same as the one for CPPs mentioned in
the previous paragraph, and the polyhedral faces of our numerical
solutions come close to the bound ($n=5,8,11$) or coincide with it
($n=4,6,7,10,12$).  The faces of the polyhedron associated with the
Majorana representation might therefore hold the key to a proof for
$2n-4$ being the upper bound on the number of CPPs for all symmetric
states (with only the Dicke states excluded).

The case of ten qubits seems to be the first one where the maximally
entangled symmetric state cannot be cast as a positive state. For
$n=10,11,12$, our candidates for maximal entanglement are real states,
so the question remains whether the maximally entangled state can
always be cast with real coefficients.  We consider this unlikely,
firstly because of the higher amount of MP freedom in the general case
and secondly because many of the solutions to the classical problems
cannot be cast as real MP distributions for higher $n$.  For Thomson's
problem, the first distribution without any symmetry (and thus no
representation as a real state) arises at $n=13$, and for T\'{o}th's
problem at $n=15$.

Upper and lower bounds on the maximal entanglement of symmetric states
have already been discussed in \sect{positive_and_symm}, with a new
proof for the upper bound given in \app{normalization_bloch}.
Stronger lower bounds can be computed from the known solutions to
T\'{o}th's and Thomson's problems by translating their point
distribution into the corresponding symmetric state and determining
its entanglement.  The diagram in \fig{e_graph} displays the
entanglement of our numerical solutions, together with all bounds.

\begin{figure}
  \begin{center}
    \begin{overpic}[scale=1.1]{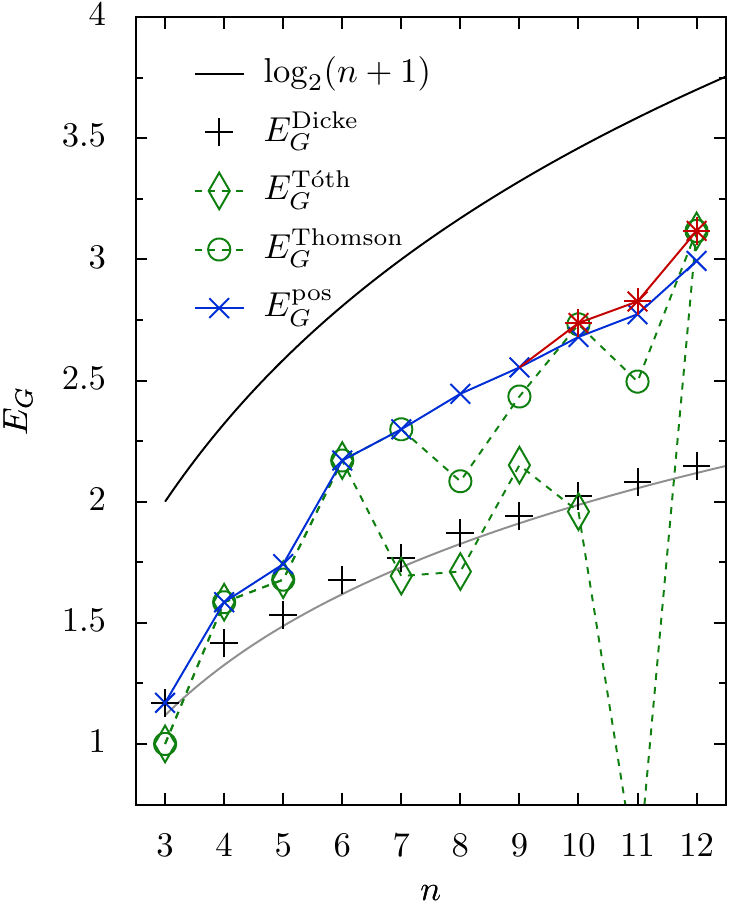}
    \end{overpic}
  \end{center}
  \caption{\label{e_graph} (color online) Scaling of maximal symmetric
    entanglement with the number of qubits $n$.  The upper bound is
    represented by a black line, while the most entangled Dicke states
    form the lower bound.  Their Stirling approximation is displayed
    as a grey line.  The most entangled positive symmetric states
    found are denoted by blue crosses. For $n = 10-12$, the best
    candidates cannot be cast with positive coefficients and are
    depicted as red stars.  The solutions of T\'{o}th's and Thomson's
    problems readily provide lower bounds, displayed as dashed green
    lines.  The solutions of Thomson's problem are generally more
    highly entangled than those of T\'{o}th's problem.}
\end{figure}

For $n > 5$ qubits, the maximally entangled state cannot be symmetric
or turned into one by LOCC, because the lower bound on general states
is higher than the upper bound of symmetric states. For $n=3$, the
maximally entangled state (W state) is demonstrably symmetric, but for
$n = 4,5$ the numerical solutions for symmetric states have less
entanglement than the lower bound of general states. This would imply
that $n$ qubit maximally entangled states can be symmetric if, and
only if, $n \leq 3$.

\section{Conclusion}\label{conclusion}

In this paper, we have investigated the maximally entangled state
among symmetric quantum states of $n$ qubits. By visualizing symmetric
states through the Majorana representation and with the help of
analytical and numerical methods, strong candidates for the maximally
entangled symmetric states of up to 12 qubits were found. A comparison
with the extremal distributions of T\'{o}th's and Thomson's problems
shows that, in some cases, the optimal solution to Majorana's problem
coincides with that of the two classical problems, but in other cases
it significantly differs.

Lower and upper bounds show that the maximal entanglement of
permutation-symmetric qubit states scales between $\Order (\log
\sqrt{n})$ and $\Order (\log n)$ with the number of qubits $n$.  With
respect to MBQC, these results indicate that, although
permutation-symmetric states may not be good resources for
deterministic MBQC, they may be good for stochastic approximate MBQC
\cite{Nest07,Mora10}. It also gives bounds on how much information can
be discriminated locally, for which explicit protocols are known in
some cases (in particular for Dicke states)
\cite{Hayashi06,Hayashi08}.

We remark that, due to the close relationship of distance-like
entanglement measures \cite{Hayashi06}, the results for the geometric
measure give bounds to the robustness of entanglement and the relative
entropy of entanglement, which can be shown to be tight in certain
cases of high symmetry \cite{Hayashi08,Markham10}.

We finally note that a similar study has been carried out, although
from a different perspective, in search of the `least classical' state
of a single spin-$j$ system \cite{Giraud10} (which they call the
`queens of quantumness'). There the Majorana representation is used to
display spin-$j$ states, through the well-known isomorphism between a
single spin-$j$ system and the symmetric state of $n=2j$ spin-$1/2$
systems. In this context, the most `classical' state is the spin
coherent state, which corresponds exactly to a symmetric product state
in our case (i.e. $n$ coinciding MPs). The problem of \cite{Giraud10}
is similar to ours in that they look for the state `furthest away'
from the set of spin coherent states. However, different distance
functions are used, so the optimization problem is subtly different
and again yields different solutions. It is in any case interesting to
note that our results also have interpretations in this context and
vice versa.

\begin{acknowledgments}

  The authors thank S.~Miyashita, S.~Virmani, A.~Soeda and
  K.-H.~Borgwardt for very helpful discussions.  This work was
  supported by the National Research Foundation \& Ministry of
  Education, Singapore, and the project `Quantum Computation: Theory
  and Feasibility' in the framework of the CNRS-JST Strategic
  French-Japanese Cooperative Program on ICT. MM thanks the `Special
  Coordination Funds for Promoting Science and Technology' for
  financial support.

  \emph{Note added.} During the completion of this manuscript, we
  became aware of very similar work that also looks at the maximum
  entanglement of permutation-symmetric states using very similar
  techniques \cite{Martin10}.
\end{acknowledgments}

\appendix

\section{Upper bound on symmetric entanglement}
\label{normalization_bloch}

A symmetric $n$ qubit state can be written as
\begin{equation*}
  | \psi \rangle = \sum_{k = 0}^{n}
  a_k \E^{\I \alpha_k} | S_k \rangle \enspace ,
\end{equation*}
with $a_k \in \mathbb{R}$, $\alpha_k \in [0,2 \pi)$ and the
normalization condition $\sum_{k} a_k^2 = 1$.  Writing the closest
product state as $| \lambda \rangle = | \sigma \rangle^{\otimes n}$
with $| \sigma \rangle = \co_{\theta} | 0 \rangle + \E^{\I \varphi}
\si_{\theta} | 1 \rangle$, we obtain
\begin{equation}\label{calc1}
  \langle \lambda | \psi \rangle = \sum_{k = 0}^{n}
  \E^{\I (\alpha_k - k \varphi )} a_k \co_{\theta}^{n-k}
  \si_{\theta}^{k} \sqrt{{\tbinom{n}{k}}} \enspace .
\end{equation}
Using the set of qubit unit vectors $\mathcal{H}_1$ and the uniform
measure over the unit sphere $d \mathcal{B}$, the squared norm of
\eq{calc1} can be integrated over the Majorana sphere:
\begin{equation}
  \int\limits_{| \sigma \rangle \in \mathcal{H}_1 }
  \vert \langle \lambda | \psi \rangle \vert^2 d \mathcal{B} =
  \int\limits_{0}^{2 \pi} \int\limits_{0}^{\pi}
  \vert \langle \lambda | \psi \rangle \vert^2 \sin \theta \,
  \D \theta \D \varphi \enspace .
\end{equation}
Taking into account that $\int_{0}^{2 \pi} \E^{\I m \varphi} \D
\varphi = 0$ for any integer $m \neq 0$, one obtains
\begin{subequations}\label{calc_mean}
  \begin{align}
    & \int\limits_{0}^{2 \pi} \int\limits_{0}^{\pi}
    \left[ \: \sum_{k = 0}^{n} a_k^2 \co_{\theta}^{2(n-k)}
      \si_{\theta}^{2k} {\binom{n}{k}} \right]
    \sin \theta \, \D \theta \D \varphi \enspace ,
    \label{calc_mean_1} \\
    & = 2 \pi \sum_{k = 0}^{n}
    a_k^2 {\binom{n}{k}} \int\limits_{0}^{\pi}
    \co_{\theta}^{2(n-k)}
    \si_{\theta}^{2k}
    \sin \theta \, \D \theta \enspace ,
    \label{calc_mean_2} \\
    & = 4 \pi \sum_{k = 0}^{n} a_k^2 {\binom{n}{k}}
    \frac{\Gamma (k+1) \Gamma (n-k+1)}{\Gamma (n+2)} \enspace ,
    \label{calc_mean_4} \\
    & = 4 \pi \sum_{k = 0}^{n} a_k^2 \frac{1}{n+1} =
    \frac{4 \pi}{n+1} \enspace .
    \label{calc_mean_5}
  \end{align}
\end{subequations}
The equivalence of Equations \eqsimple{calc_mean_2} and
\eqsimple{calc_mean_4} follows from the different definitions of the
Beta function \cite{Abramowitz}.  Since the mean value of $|\langle
\lambda | \psi \rangle|^2$ over the Majorana sphere is $4 \pi /
(n+1)$, it follows that $\text{G} ( | \psi \rangle )^2 \geq 1/(n+1)$,
or $E_{\text{G}} ( | \psi \rangle ) \leq \log_2 (n+1)$, for any
symmetric $n$ qubit state.

This result was first shown by R.~Renner in his PhD thesis
\cite{Renner}, using a similar proof that employs an explicit
separable decomposition of the identity over symmetric subspace.  The
same proof as ours was independently found by J. Martin \etal
\cite{Martin10}.

\section{Proof of theorem \ref{maj_max_pos_zero}}
\label{max_cpp_number}

Class (b): We consider states that have a Z-axis rotational symmetry
with minimal rotational angle $\varphi = 2 \pi / m$, $m \in
\mathbb{N}$ and $1<m \leq n$.  \Fig{symm_example} shows an example for
$m=5$.  Due to the rotational Z-axis symmetry and the reflective
symmetry imposed by Theorem \ref{maj_real}, the MPs are restricted to
specific distribution patters.  An arbitrary number of MPs can lie on
each of the poles, with the remaining MPs equidistantly aligned along
horizontal circles.  The figure shows the two principal types of
horizontal circles that can exist. The upper one is the basic type for
positive states of five qubits, and the lower one a special case where
two basic circles are intertwined at azimuthal angle $\pm \vartheta$
from the position of the single basic circle, respectively. All
conceivable horizontal circles of MPs can be decomposed into these two
principal types.

\begin{figure}
  \begin{center}
    \begin{overpic}[scale=.5]{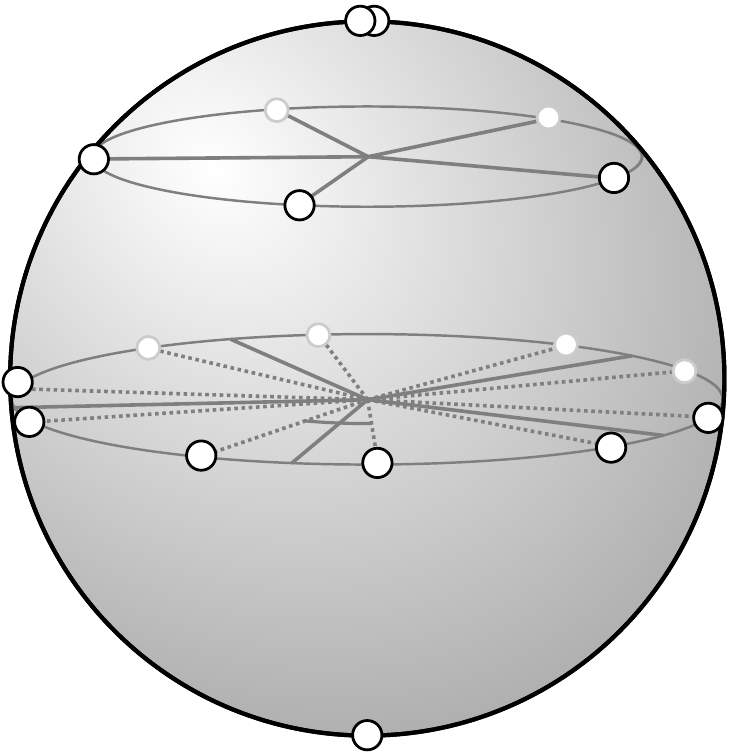}
      \put(36.3,39.3){{\scriptsize $\vartheta$}}
      \put(43.8,39.0){{\scriptsize $\vartheta$}}
    \end{overpic}
  \end{center}
  \caption{\label{symm_example} An example of a positive state of 18
    qubits with a rotational symmetry around the Z-axis with angle
    $\varphi = 2 \pi / 5$. Two MPs lie on the north pole, one on the
    south pole, five on a single basic circle and ten on two
    intertwined basic circles.}
\end{figure}

According to \eq{bloch_product}, any CPP $\vert \sigma \rangle$
maximizes the function $\prod_{i=1}^{n} \, | \langle \sigma | \phi_i
\rangle |$, where the $\vert \phi_i \rangle$ are the MPs.  From Lemma
\ref{lem_positive} it follows that there must be at least one positive
$\vert \sigma \rangle = \co_{\theta} |0\rangle + \si_{\theta}
|1\rangle$. We first derive the maximum number of positive CPPs, and
from these the upper bound for the total number of CPPs can be
immediately obtained with the help of Lemma \ref{cpp_mer}.

For a MP distribution with $k$ MPs on the north pole, $l$ MPs on the
south pole and the remaining $n-k-l$ MPs on horizontal circles, the
function to maximize is
\begin{equation*}
  f(\theta) = \langle \sigma | 0 \rangle^{k} \langle \sigma
  | 1 \rangle^{l} \prod_{r} h_1 (\theta_{r}) \prod_{s} h_2
  (\vartheta_{s} , \theta_{s}) \enspace ,
\end{equation*}
where $h_1 (\theta_{r})$ represents the factors contributed by a
single basic circle with $m$ MPs at inclination $\theta_{r}$, and $h_2
(\vartheta_{s}, \theta_{s})$ represents the factors contributed by two
basic circles with $2m$ MPs intertwined at azimuthal angles $\pm
\vartheta_{s}$, and inclination $\theta_{s}$. It is easy to verify
that
\begin{align*}
  h_1 (\theta_{r}) & = \co_{\theta}^{m} \co_{\theta_r}^{m} +
  \si_{\theta}^{m} \si_{\theta_r}^{m} \enspace , \\
  h_2 (\vartheta_{s} , \theta_{s}) & =
  \co_{\theta}^{2m} \co_{\theta_s}^{2m} + 2 \cos
  (m \vartheta_{s}) \co_{\theta}^{m} \si_{\theta}^{m}
  \co_{\theta_s}^{m} \si_{\theta_s}^{m} + \si_{\theta}^{2m}
  \si_{\theta_s}^{2m} \: .
\end{align*}
From this it is clear that $f$ can be written in the form
\begin{equation*}
  f(\theta) = \co_{\theta}^{k} \si_{\theta}^{l} \sum\limits_{i=0}^{p}
  a_i  \co_{\theta}^{(p-i)m} \si_{\theta}^{im} = \sum\limits_{i=0}^{p}
  a_i  \co_{\theta}^{k+(p-i)m} \si_{\theta}^{l+im} \enspace ,
\end{equation*}
where the $a_i$ are positive-valued coefficients, and $p$ is the
number of basic circles ($k+l+pm=n$).  The number of zeros of
$f'(\theta)$ in $\theta \in (0, \pi)$ gives a bound on the number of
positive CPPs. The form of $f'(\theta)$ is qualitatively different for
$m=2$ and $m>2$. With the substitution $x = \tan (\theta/2)$ the
equation $f'(\theta) = 0$ for $m=2$ becomes
\begin{gather*}
  a_{0} l + \left( \sum\limits_{i=1}^{p} b_i x^{2i} \right) -
  a_{p} k x^{2p+2} = 0 \enspace , \quad \text{with} \\
  b_{i} = a_{i} (l+2i) - a_{i-1} (k+2(p-i)+2)
  \in \mathbb{R} \enspace .
\end{gather*}
This is a real polynomial in $x$, with the first and last coefficient
vanishing if no MPs exist on the south pole ($l=0$) and north pole
($k=0$), respectively.  Descartes' rule of signs states that the
number of positive roots of a real polynomial is at most the number of
sign differences between consecutive nonzero coefficients, ordered by
descending variable exponent. From this and the fact that the codomain
of $x$ is $\mathbb{R}^{+}$, we obtain the result that for $m=2$ there
are at most $p-1$, $p$ or $p+1$ extrema of $f(\theta)$ lying in
$\theta \in (0,\pi)$, depending on whether $k$ and $l$ are zero or
not.

For $m>2$, we obtain the analogous result
\begin{align*}
  a_{0} l + {} & \Bigg( \sum\limits_{i=1}^{p} - c_i x^{im-(m-2)} +
  d_i x^{im} \Bigg) - a_{p} k x^{pm+2} = 0 \enspace , \\
  & \text{with} \quad c_{i} = a_{i-1} (k+(p-i)m+m)
  \in \mathbb{R}^{+} \enspace , \\
  & \text{and} \quad \: d_{i} = a_{i} (l+im) \in \mathbb{R}^{+}
  \enspace .
\end{align*}
From Descartes' rule of signs, we find that there exist $2p-1$, $2p$
or $2p+1$ extrema of $f(\theta)$ in $\theta \in (0, \pi)$, depending
on whether $k$ and $l$ are zero or not.

With these results it is easy to determine the maximum number of
global maxima of $f(\theta)$, which are the positive CPPs. Case
differentiations have to be done with regard to $m=2$ or $m>2$,
whether $k$ and $l$ are zero or not and whether $p$ is even or odd.
Due to the rotational Z-axis symmetry, the non-positive CPPs can be
immediately obtained. For any positive CPP not lying on a pole, there
are $m-1$ other, non-positive CPPs lying at the same inclination
(cf. Lemma \ref{cpp_mer}).  For $m=2$, the maximum possible number of
CPPs is $(n/2) + 1$ ($n$ even) or $(n+1)/2$ ($n$ odd). This is
significantly less than that in the general case $m>2$, where the
maximum number of CPPs is $2n-4$. Interestingly, this bound decreases
to $n$ if at least one of the two poles is free of MPs.

Class (c): All MPs of a positive state must either lie on the positive
meridian or form complex conjugate pairs (cf. Lemma
\ref{maj_real}). From this the optimization function is easily derived
as
\begin{equation*}\label{ffform}
  f(\theta) = \sum\limits_{i=0}^{n} a_i  \co_{\theta}^{n-i}
  \si_{\theta}^{i} \enspace ,
\end{equation*}
with real $a_i$. Calculating $f'(\theta)$ yields the condition for the
extrema:
\begin{align*}
  a_{1} + {} & \left( \sum\limits_{i=1}^{n-1} b_i x^{i} \right)
  - a_{n-1} x^{n} = 0 \enspace , \\
  & \text{with} \quad b_{i} = a_{i+1} (i+1) - a_{i-1} (n-i+1)
  \enspace .
\end{align*}
From this, the maximum number of CPPs can be derived with Descartes'
rule. All CPPs are now restricted to the positive meridian and the
poles, yielding at most $(n+3)/2$ CPPs for odd $n$ and $(n+2)/2$ for
even $n$.  \hfill $\square$

\newpage

\section{Table of entanglement values}\label{entanglement_table}

\begin{table}[h!]
  \caption{\label{enttable}
    The table lists the known ($n= 2,3$) and numerically
    determined ($n > 3$) values of the maximal entanglement of symmetric
    $n$ qubit states. The left column lists the extremal entanglement
    among positive symmetric states, and, where more entangled non-positive
    symmetric states are known, they are displayed in the right column.}
  \begin{tabular}{|c||c|c|} \hline
    n & $E_{\text{G}}^{\text{pos}}$ & $E_{\text{G}}$ \\ \hline \hline
    2 & $1$ &  \\
    3 & $2 \log_2 3 - 2 \approx 1.170$ &  \\
    4 & $\log_2 3 \approx 1.585$ &  \\
    5 & $\approx 1.742 \: 268 \: 948$ \footnotemark[1] &  \\
    6 & $2 \log_2 3 - 1 \approx 2.170$ &  \\
    7 & $\approx 2.298 \: 691 \: 396$ \footnotemark[1] &  \\
    8 & $\approx 2.445 \: 210 \: 159$ \phantom{$^{a}$} &  \\
    9 & $\approx 2.553 \: 960 \: 277$ \footnotemark[1]  &  \\
    10 & $\approx 2.679 \: 763 \: 092$ \phantom{$^{a}$} &
    $\approx 2.737 \: 432 \: 003$ \\
    11 & $\approx 2.773 \: 622 \: 669$ \phantom{$^{a}$} &
    $\approx 2.817 \: 698 \: 505$ \\
    12 & $\approx 2.993 \: 524 \: 700$ \phantom{$^{a}$} &
    $\log_2 (243/28) \approx 3.117$ \\ \hline
  \end{tabular}
  \footnotetext[1]{The analytic form is known, but is of a complicated form.}
\end{table}

\end{document}